%% file: unf-ortho.tex
\documentclass{elsarticle}
\usepackage[colorlinks=true,linkcolor=black,citecolor=black,urlcolor=black]{hyperref}

\usepackage{amsmath,amsfonts,amssymb,bm}
\usepackage{cleveref}
\usepackage{latexsym}
\usepackage{enumerate, enumitem}
\usepackage{array,multirow}
\usepackage{xcolor}
\usepackage{thm-restate}
\usepackage{wrapfig}

\setlist[itemize]{leftmargin=*, itemsep=0.2em, label={\scriptsize$\bullet$}}

\newcommand{\notyet}[1]{}
\newcommand{\squeezelist}{\setlength{\itemsep}{0pt}}

\newcommand{\ABox}{
\raisebox{3pt}{\framebox[6pt]{\rule{6pt}{0pt}}}
}
\newenvironment{proof}{{\bf Proof:}}{\hfill\ABox}

\newtheorem{theorem}{{\bf Theorem}}

\newtheorem{lemma}[theorem]{Lemma}

\newtheorem{definition}[theorem]{Definition}

\def\O{{\mathcal O}}

\def\I{\bm{\mathcal{I}}}
\def\P{\mathnormal P}
\def\I{\bm{I}}
\def\clip{clip}
\def\beam{beam}
\def\beams{beams}

\def\beams{beams}

\crefname{theorem}{Theorem}{Theorems}
\crefname{lemma}{Lemma}{Lemmas}
\crefname{proposition}{Proposition}{Propositions}
\crefname{corollary}{Corollary}{Corollaries}
\crefname{figure}{Figure}{Figures}
\crefname{table}{Table}{Tables}
\crefname{property}{Property}{Properties}
\crefname{section}{Section}{Sections}
\crefname{definition}{Definition}{Definitions}

\newcommand\hide[1]{}

\begin{document}


\title{Unfolding Polycubes with Orthogonally Convex Layers\tnoteref{tnote}}
\tnotetext[tnote]{Preliminary ideas explored in this work were first introduced in~\cite{Damian-Meijer-2004-orthostacks}.}

\author{Mirela Damian\fnref{md}\corref{cor1}}
\ead{mirela.damian@villanova.edu}
\author{Henk Meijer\fnref{hm}}
\ead{hendrikus.meijer@gmail.com}
\cortext[cor1]{Corresponding author}

\affiliation[md]{organization={Villanova University},
            city={Villanova},
            state={PA},
            country={USA}}
\affiliation[hm]{organization={University College Roosevelt},
            city={Middelburg},
            country={Netherlands}}

\begin{abstract}
A \emph{polycube} is an orthogonal polyhedron composed of unit cubes glued together along entire faces, and homeomorphic to a sphere. A \emph{layer} of a polycube refers to the portion lying between two horizontal cross-sections spaced one unit apart. We present an unfolding algorithm that flattens any polycube with orthogonally convex layers into a single, non-overlapping planar piece. The algorithm makes cuts only along cube edges—that is, it is an \emph{edge unfolding}.
\end{abstract}


\begin{keyword}
unfolding  algorithms \sep orthostacks \sep polycubes \sep orthogonally convex 
\vspace{0.5em}
\MSC 52C45 \sep 05C10
\end{keyword}

\maketitle              

\vspace{-1em}
\section{Introduction}
An \emph{unfolding} of a polyhedron involves cutting its surface and flattening
it into a single, non-overlapping planar piece called an \emph{unfolding net}. 
\emph{Edge unfoldings} limit cuts to the edges of the polyhedron, whereas 
\emph{general unfoldings} allow cuts anywhere, including across face interiors.

The problem of unfolding polyhedra has been extensively studied. 
For non-convex polyhedra, it has been shown that edge cuts alone are not sufficient to
guarantee an unfolding~\cite{Bern-Demaine-Eppstein-Kuo-Mantler-Snoeyink-2003,BDDLOORW1998},
and it remains open whether all non-convex polyhedra admit a general unfolding.
In contrast, every convex polyhedron does have a general unfolding~\cite[Ch.~22]{Demaine-O'Rourke-2007}, but it is still open whether all convex polyhedra have an edge unfolding. 

Research on unfolding non-convex objects has primarily focused on \emph{orthogonal} polyhedra, whose edges and faces meet at right angles.
However, since orthogonal polyhedra cannot always be unfolded using edge cuts alone~\cite{BDDLOORW1998}, 
unfolding algorithms typically introduce additional non-edge cuts. These extra cuts occur along a  refined grid structure created by intersecting the polyhedron's surface with axis-aligned planes through each vertex. The surface is thereby divided into rectangular faces, which can be further subdivided into a finer $a \times b$ grid of rectangular subfaces, for positive integers $a, b \ge 1$. 
Cuts are permitted along any of these introduced grid lines, enabling more flexible unfolding strategies.

For orthogonal polyhedra homeomorphic to a sphere, progressive results have reduced the amount of refinement required for unfolding.  
The initial method~\cite{Damian-Flatland-O'Rourke-2007-epsilon} required exponential refinement, later improved to quadratic refinement~\cite{Damian-Demaine-Flatland-2014-delta} and then to linear refinement~\cite{Chang2015}. 
These techniques were further extended to unfold genus-$2$ orthogonal polyhedra using 
linear refinement~\cite{Damian-Demaine-Flatland-2017-O'Rourke-genus2}.
Constant refinement has been achieved only for specialized subclasses of orthogonal polyhedra, including orthostacks using $1 \times 2$ refinement~\cite{BDDLOORW1998},  
Manhattan Towers using $4 \times 5$ refinement~\cite{Damian-Flatland-O'Rourke-2005-manhattan}, and
polycube trees using $4 \times 4$ refinement~\cite{Damian-Flatland-21}. 
Unfolding techniques using $1 \times 1$ refinement 
exist for orthotubes~\cite{BDDLOORW1998,DK-2023},
well-separated orthotrees~\cite{Damian-Flatland-Meijer-O'Rourke-2005-orthotrees}, 
orthostacks with rectangular faces~\cite{CST-2012}, and orthostacks with rectangular slabs~\cite{Klara-2024}. 
Also, single-layer polycubes with sparse cubic holes~\cite{Liou-Poon-Wei-2014-onelayer}
and general cubic holes~\cite{Josef-22} can be unfolded using $1 \times 1$ refinement. 

%

This paper explores edge unfoldings of polycubes (problem 64 from The Open Problems Project~\cite{open-problems}). 
A \emph{polycube}  is an orthogonal polyhedron composed of unit 
cubes\footnote{If rectangular boxes are used in place of unit cubes, this is known as an \emph{orthostack}.} 
glued together along entire faces, homeomorphic to a sphere.  
For polycubes, all edges of the constituent cubes are available for cuts.
Conceptually, a polycube can be visualized as a stack of layers, each forming a unit-height prism with a polyomino base. 
This paper considers polycubes with \emph{orthogonally convex layers}, 
where each layer intersects a line parallel to a coordinate axis in either a single line segment or not at all. Our main result is as follows: 

\begin{theorem}[Main result]
Any polycube with orthogonally convex layers has an edge unfolding. 
\end{theorem}
%


\vspace{-1.5em}
\section{Terminology}
\vspace{-0.5em}
Throughout this paper, $\O$ refers to a polycube with orthogonally convex layers. 
We use the term  \emph{face} to refer to a face of $\O$ and the term \emph{cell} to specifically mean one of the individual unit cube faces 
that make up the larger faces.
Therefore, a face is a maximal set of edge-connected, coplanar cells. 
Two non-overlapping surface pieces of $\O$ are considered \emph{adjacent} if their boundaries share at least one cell edge. 

In space, we use the term \emph{vertical} to refer to the $z$-direction and \emph{horizontal} to refer to the $x$- and $y$-directions. 
Faces are categorized by the direction of their outward normals: right is $+x$; left is $-x$; front is $+y$; back is $-y$; top is $+z$; and bottom is $-z$. %
Let $z_0, z_1, \ldots, z_m$ be the distinct $z$-coordinates of the vertices of $\O$. Let $z = z_i$ denote the $i$-\emph{plane}. 
We define the $i$-\emph{band} as the collection of vertical cells parallel to the $z$-axis that lie between the $(i-1)$-plane and the $i$-plane, 
for $i = 1, 2, \ldots, m$. 
By the definition of a polycube, each $i$-band is the boundary of an extruded simple orthogonal polygon and is therefore connected. The layer $\O_i$ is the portion of $\O$ bounded by the $i$-band, the $(i-1)$-plane, and the $i$-plane, forming a prism with an orthogonal polygon as its base. Thus, $\O$ can be viewed as a collection of orthogonally convex layers $\O_1, \O_2, \ldots, \O_m$ stacked in this order along the positive $z$-axis. 
See~\cref{fig:defs-1}(a) for an example of a polycube with 7 orthogonally convex layers.  

\begin{figure}[htp]
    \centering
    \includegraphics[page=1,width=0.98\textwidth]{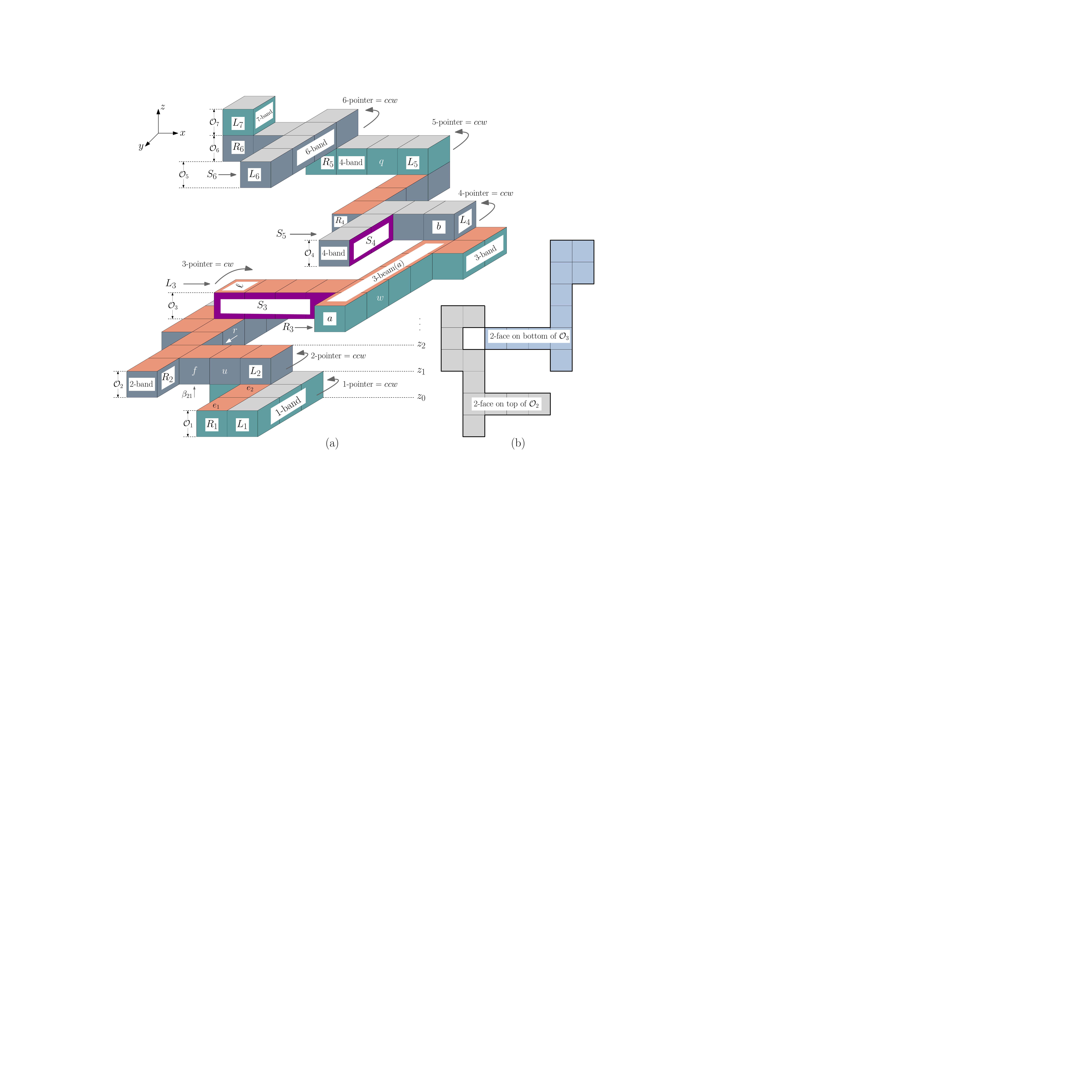}
    \caption{(a) Polycube $\O$ with 7 orthogonally convex layers;  visited $i$-band segments are delimited by $L_i$ and $R_i$, and $i$-bridges  are highlighted in brown, for $i = 1, \ldots, 6$ (other labels will be discussed in~\cref{sec:unfolding}.)    
    (b) The $2$-faces of $\O$.}
    \label{fig:defs-1}
\end{figure}

An $i$-\emph{face} is a face of $\O$ that lies in the $i$-plane. 
Thus, the surface of a polycube is composed of vertical $i$-bands and horizontal $i$-faces. 
Note that an $i$-face may be positioned either on top of $\O_i$ or on the bottom of $\O_{i+1}$: if on top of $\O_i$, it is called a \emph{top} face; 
if on the bottom of $\O_{i+1}$, it is called a \emph{bottom} face. 
See~\cref{fig:defs-1}(b) for an example. 
Note that a top (bottom) face has the interior of $\O$ below (above) it. 
Similar definitions apply to cells. 

For any straight band segment $S$ adjacent to an $i$-face along an edge segment $e$, $i$-\emph{\beam}$(S)$ is defined as the  
$i$-face piece illuminated by rays emitted from $e$ in the direction orthogonal to $S$.
When talking about $i$-\beam($a$), for a single band cell $a$, we call $a$ an \emph{anchor} to suggest the possibility of fastening (or ``anchoring'') the beam to $a$ during unfolding. We say that $i$-\beam($a$) is \emph{anchored} on $a$, or more broadly, on the band containing $a$. 
Note that $i$-\emph{\beam}$(a)$ has two parallel \emph{anchors}: one is $a$, and the other is the band cell $b$ that is parallel to $a$ and 
adjacent to $i$-\emph{\beam}$(a)$. 
If $b$ lies on the $i$-band, it is called \emph{opposite} to $a$; otherwise, $b$ lies on the $(i+1)$-band.
For example, in~\cref{fig:defs-1}(a), $3$-beam($a$) is anchored on cells $a$ and $b$, which are on the $3$-band and the $4$-band, respectively. 
Note that, if an $i$-beam is anchored on two cells $a$ and $b$, $i$-\beam($a$) is equivalent to $i$-beam($b$). Two beams are considered \emph{parallel} if their anchors are parallel. 

\vspace{-0.2em}
\begin{lemma} [Restatement of Lemma 2 from~\cite{CST-2012}]
The perimeter of any $i$-face is divided into two contiguous components: one adjacent to the $i$-band and one adjacent to the $(i+1)$-band. 
Moreover, every $i$-face has two opposite cell edges on its boundary, one on the 
$i$-band and one on the $(i+1)$-band. 
\label{lem:opposite-edges}
\end{lemma}

\vspace{-1.5em}
\section{Band Segments and Bridges}
The main idea of our unfolding algorithm is simple. For each $i$-band, we identify two band cells, $L_i$ and $R_i$, whose names reflect their roles as the leftmost and rightmost 
$i$-band cells in the unfolding net. We also identify an unfolding direction, $i$-$pointer \in \{cw, ccw\}$, to guide the unfolding of $\O_i$, where $cw$ and $ccw$ denote clockwise and counterclockwise, respectively, as viewed from $+z$. 

Consider any two $i$-band cells, $a$ and $b$. We define $i$-$band[a, b]$ as the \emph{closed} $i$-band segment from $a$ to $b$, inclusive, following the direction of the $i$-pointer. Similarly, $i$-$band(a, b)$ denotes the \emph{open} segment from $a$ to $b$, excluding both endpoints $a$ and $b$. By these definitions, the entire $i$-band is the union of $i$-band$[a, b]$ and $i$-band$(b, a)$. 

The unfolding algorithm cuts along the perimeter of each $i$-band$[L_i, R_i]$ and unfolds it as a horizontal rectangular piece in the plane, with $L_i$ to the left and $R_i$ to the right. For consecutive layers, the band segments $i$-band$[L_i, R_i]$ and $(i+1)$-band$[L_{i+1}, R_{i+1}]$ are selected so that $R_i$ and $L_{i+1}$ are parallel. The beams $i$-\beam$(R_i)$ and $i$-\beam$(L_{i+1})$ delimit a connected region of the $i$-face---called the $i$-\emph{bridge}---that links the two band segments,
with $i$-\beam$(R_i)$ to the left and $i$-\beam$(L_{i+1})$ to the right in the unfolding. 
The $i$-band cells outside of $i$-band$[R_i, L_i]$ are handled separately.
With few exceptions, $i$-\beams\ are attached above and below appropriate anchors in the band unfoldings. 

\setlength{\intextsep}{-0.2em}%
\setlength{\columnsep}{1em}%
\begin{wrapfigure}{r}{0.32\textwidth}
\raisebox{-1.5em}{\includegraphics[page=2,width=0.31\textwidth]{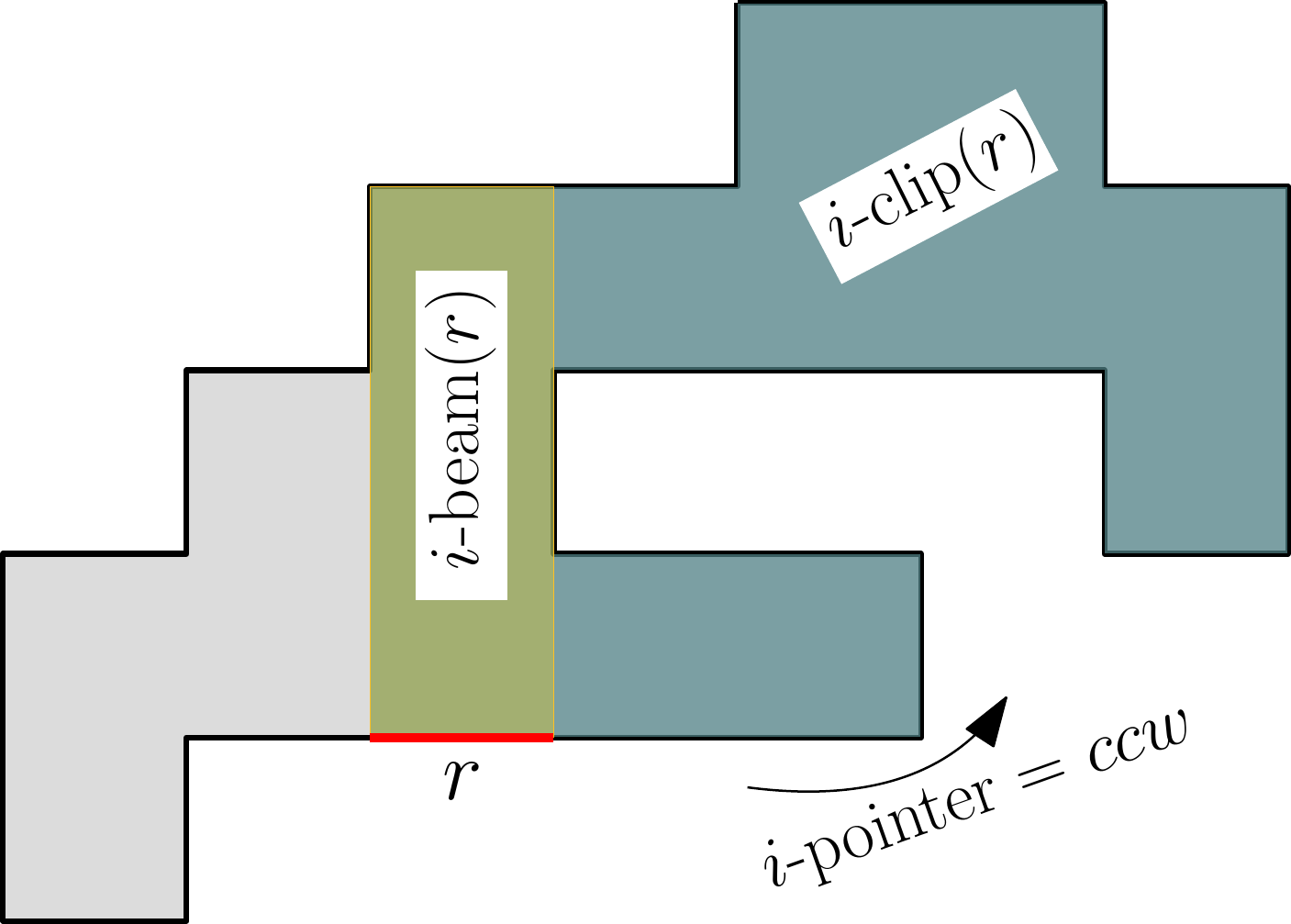}}
 \caption{Top view of an $i$-face; $i$-\clip($r$) is the dark-shaded piece.}
\label{fig:defs-clip}
\end{wrapfigure}

\begin{definition} [{\bf $i$-\clip}]
For a fixed $i$-band cell $r$, $i$-\clip($r$) is the portion of the $i$-face starting at (and including)  $i$-\beam($r$) and extending in the direction of the $i$-pointer. If $i$-\beam($r$) is empty, then $i$-\clip($r$) is also empty. 
\label{def:iclip}
\end{definition}

\noindent
This definition is depicted in~\cref{fig:defs-clip}. 
Note that $i$-\clip($r$) is connected, since it includes $i$-\beam($r$). 

\subsection{Selecting Band Segments and Bridges}
\label{sec:bands-bridges}
In this section, we show how to select $L_i$, $R_i$, the $i$-bridge, and the $i$-pointer, for increasing $i \ge 1$. 
We start by setting both $1$-pointer and $2$-pointer to $ccw$. 
Let $f$ be any $1$-face of $\O_1$, and let $e_1$ and $e_2$ be two opposite grid edges of $f$, adjacent to the $1$-band and the $2$-band, respectively. By~\cref{lem:opposite-edges}, such edges exist. 
We set the following: $R_1$ is the $1$-band cell adjacent to $e_1$; the $1$-bridge is $1$-\beam($R_1$); $L_1$ is the $1$-band cell adjacent to $R_1$ in the direction of the $1$-pointer (i.e., $ccw$); and $L_2$ is the $2$-band cell adjacent to $e_2$. See~\cref{fig:defs-1}(a). 

For a fixed $i > 1$, suppose that both $L_i$ and the $i$-pointer for traversing the $i$-band are fixed. This assumption holds as we  handle $i$-bands sequentially, for increasing $i$. We now show how $R_{i}$, $L_{i+1}$, the $i$-\emph{bridge} connecting the two, and  the $(i+1)$-pointer are selected by our unfolding approach. 

\begin{definition} [{\bf $R_i$, $L_{i+1}$, $i$-bridge, $(i+1)$-pointer}] 
Starting at $L_i$, walk \\ around the $i$-band in the direction of the $i$-pointer. Let $R_i$ to be the  \emph{last} $i$-band cell encountered such that $i$-\clip$(R_i)$ is either empty or adjacent to an $(i+1)$-band cell parallel to $R_i$. 

\medskip
\noindent
If $i$-\clip($R_i$) is empty, set $L_{i+1}$ to the $(i+1)$-band cell adjacent to $R_i$, the $i$-bridge to empty, and the $(i+1)$-pointer equal to the $i$-pointer.  

\medskip
\noindent
Otherwise, partition $i$-\clip($R_i)$ into beams parallel to $i$-\beam$(R_i)$: 
\begin{enumerate}
\squeezelist
\item [$(a)$] Set $L_{i+1}$ to the $(i+1)$-band cell parallel to $R_i$ and adjacent to $i$-\clip($R_i$) that  minimizes the number of beams between $i$-\beam($R_i$) and $i$-\beam($L_{i+1}$); in case of ties, 
choose the one closest to $R_i$ in Manhattan distance. 
\item [$(b)$] Set the $i$-bridge to the $i$-face piece delimited by $i$-\beam$(R_i)$ and $i$-\beam$(L_{i+1})$. 
\item [$(c)$] If $R_i$ and $L_{i+1}$ share the same normal, set the $(i+1)$-pointer to the $i$-pointer; otherwise, reverse it. 
\end{enumerate}
A cell in the $i$-band is \emph{visited} if it is in $i$-band$[L_i, R_i]$, and \emph{unvisited} otherwise.
\label{def:ri}
\end{definition}

\vspace{-1em}
\noindent
For an illustration of~\cref{def:ri}, see~\cref{fig:defs-1,fig:ortho-ex-complete}: 
all bridges are marked in brown; the $1$-bridge is a single beam parallel to the $y$-axis;
the 2-bridge consists of five top beams parallel to the $x$-axis; 
the 3-bridge spans the entire 3-face with six top beams parallel to the $x$-axis; 
the 4-bridge is a single top beam parallel to the $y$-axis;
the $5$-bridge is a single bottom beam parallel to the $y$-axis;
and the $6$-bridge is empty. 

\medskip
The following lemmas capture key structural properties that guide the unfolding process. 
To keep the focus on the unfolding algorithm, we defer their proofs to~\ref{sec:appendix-lemmas}. 

\begin{restatable}{lemma}{liriexist}
\label{lem:ri}
$R_i$ and $L_{i+1}$ exist and are uniquely defined. 
\end{restatable}

\vspace{-1em}
\begin{restatable}{lemma}{beamends}
\label{lem:beam-ends}
Any beam on the surface of $\O_i$ is anchored on the $i$-band.
\end{restatable}

\vspace{-1em}
\begin{restatable}{lemma}{liriout}
\label{lem:liri-out}
For every unvisited $i$-band cell $u$, $i$-\clip($u$) is non-empty, lies on top of $O_i$, and is not adjacent to any $(i+1)$-band cell parallel to $u$.
\end{restatable}

\vspace{-1em}
\begin{restatable}{lemma}{bridgebot}
\label{lem:bridge-bot}
If the $i$-bridge is on a bottom face, it consists of a single beam. 
\end{restatable}

\vspace{-1em}
\begin{restatable}{lemma}{lemstrips}
\label{lem:strips}
For every unvisited $i$-band cell $u$ parallel to $R_i$, $i$-\beam$(u)$ is non-empty and is anchored on a visited $i$-band cell. 
\end{restatable}

\vspace{-1em}
\begin{restatable}{lemma}{frontallback}
\label{lem:front-allback}
All $i$-band cells with normals opposite to that of $R_i$ are visited. However, $L_i$ may not be opposite to $R_i$.
\end{restatable}

\section{Unfolding Algorithm}
\label{sec:unfolding}
The unfolding algorithm is fairly simple, but its correctness relies on intricate reasoning.
In this paper we present the algorithm itself and refer the reader to~\cite{DM24} for a 
detailed proof of correctness.

\medskip
\noindent
Our unfolding procedure for $\O$ consists of four stages: 

\medskip
{\bf Stage 1:} Unfolding visited band segments and bridges (\cref{sec:bands}) 

\vspace{0.4em}
{\bf Stage 2:} Unfolding the top surface of $\O$ (\cref{sec:top}) 

\vspace{0.4em}
{\bf Stage 3:} Unfolding the bottom surface of $\O$ (\cref{sec:bottom}) 

\vspace{0.4em}
{\bf Stage 4:} Unfolding remaining band pieces of $\O$ (\cref{sec:allbands})

\medskip
\medskip
\noindent
We say that the $i$-band is \emph{nearly-visited} if $L_i$ and $R_i$ are orthogonal and connected by a \emph{straight} $i$-band segment. For example, in~\cref{fig:defs-1}(a), the $2$-band is nearly-visited, but the $3$-band is not, since $L_3$ and $R_3$ are parallel. 

\subsection{Stage 1: Unfolding Visited Band Segments and Bridges}  
\label{sec:bands}

For each $i \ge 0$, we cut out the entire visited segment $i$-band$[L_i, R_i]$ and unfold it horizontally in the plane, with $L_i$ to the left and $R_i$ to the right. 
The $i$-bridge, delimited by $i$-\beam($R_i$) and $i$-\beam($L_{i+1}$), unfolds as a single piece 
connecting $i$-band$[L_i, R_i]$ and $(i+1)$-band$[L_{i+1}, R_{i+1}]$. See~\cref{fig:ortho-ex-complete}.

\subsection{Stage 2: Unfolding the Top Surface of $\O$}  
\label{sec:top}
We now describe the unfolding process for the top surface of each $\O_i$, for $i \ge 1$. 
Without loss of generality, we assume that the $i$-pointer is $ccw$; the $cw$ case is symmetric.

We begin by partitioning the top faces of $\O_i$ into beams parallel to $i$-\beam($R_i$). 
Consider a beam $\beta$ in this partition that is not part of the $i$-bridge. 
By~\cref{lem:strips}, $\beta$ has at least one visited anchor $r$ on the $i$-band. 
\begin{itemize}
\squeezelist
\item If $\beta$'s second anchor $q \neq r$ is not a visited $i$-band cell, we attach $\beta$ above $r$ in the unfolding net, along their shared side.
If $q$ is an unvisited $i$-band cell, we also attach $q$ above $\beta$ in the unfolding net. 
See, for example, the $5$-band cell labeled $q$ in~\cref{fig:ortho-ex-complete}.
\item If both of $\beta$'s anchors are visited $i$-band cells and one anchor is $L_i$, we attach $\beta$ above $L_i$. Otherwise, we attach $\beta$ 
above its \emph{second} visited anchor, along their shared side. 
See, for example, the $2$-beam labeled $\beta_{21}$ in~\cref{fig:ortho-ex-complete}. 
This strategy minimizes the need to relocate beams later during unfolding. 
\end{itemize}
At this stage, the unfolding net includes the entire top surface of $\O_i$ and all $i$-band cells parallel to $R_i$. 

\subsection{Stage 3: Unfolding the Bottom Surface of $\O$}  
\label{sec:bottom}%
This section describes the unfolding process for the bottom surface of each $\O_i$, for $i \ge 1$. 
As in the previous stage, we assume that the $i$-pointer is $ccw$; the $cw$ case is symmetric.

We begin by partitioning the bottom faces of $\O_i$ into beams that are parallel to $(i-1)$-\beam($L_i$) if the $i$-band is nearly-visited, and to $(i-1)$-\beam($R_i$) otherwise. As we will later see, this choice ensures that no bottom beam is anchored on two unvisited  $i$-band rectangles orthogonal to $R_i$. 

Consider a beam $\beta$ in this partition that is not part of the $(i-1)$-bridge. 
By~\cref{lem:beam-ends}, at least one of $\beta$'s anchors is on the $i$-band. 
We determine where to attach $\beta$ to the unfolding net.

\paragraph{Unfolding the bottom of $\O_i$ when the $i$-band is nearly-visited} In this case, $\beta$ is parallel  
to $(i-1)$-\beam($L_i$). By~\cref{lem:beam-ends}, at least one of $\beta$'s anchors is on the $i$-band. 
Since the $i$-band is nearly-visited, if both anchors are on the $i$-band, at least one is visited. We attach $\beta$ 
as follows:

\begin{itemize}
\squeezelist
\item If both of $\beta$'s anchors are on the $i$-band, but only one---say, $r$---is visited, attach 
$\beta$ below $r$, and $\beta$'s unvisited anchor below $\beta$. 
See, for example, the $1$-beam labeled $\beta_{22}$ in~\cref{fig:ortho-ex-complete}. 

\item If both of $\beta$'s anchors are on the $i$-band and both are visited, attach $\beta$ below the \emph{second} visited anchor (to minimize further relocations). 
See, for example, the $1$-beams labeled $\beta_{23}$ and $\beta_{24}$ in~\cref{fig:ortho-ex-complete}.
\item Otherwise, $\beta$ has one unvisited anchor $r$ on the $i$-band, and its second anchor $q$ on the $(i-1)$-band. 
By~\cref{lem:liri-out}, $q$ is visited and therefore already in the unfolding net. In this case, attach $\beta$ to $q$, and $r$ to $\beta$. 
\end{itemize}

\noindent
Additionally, for each unvisited $i$-band cell $u$ parallel to $L_i$ (if such cells exist) and
adjacent to a visited $(i-1)$-band cell $v$, attach $u$ to $v$ in the unfolding net. See, for example, the $2$-band cell labeled $u$ in~\cref{fig:ortho-ex-complete}, attached to the $1$-band. 

\begin{lemma}
If the $i$-band is nearly-visited, then at this point the unfolding net includes the entire surface of $O_i$.
\label{lem:nearly-visited}
\end{lemma}
\begin{proof}
Each beam in the partition of the top surface of $\O_i$ is attached to the net. Since their union covers the entire top surface, it is now part of the net---and the same holds for the bottom surface of $\O_i$. It remains to show that unvisited $i$-band cells are included. 
Because the $i$-band is nearly visited, $L_i$ and $R_i$ are orthogonal and connected by a straight segment $S$. Let $q \in S$ be arbitrary. 
\begin{itemize}
\squeezelist
\item If $S$ is parallel to $R_i$, then $i$-\beam($q$) 
is in the top partition of $\O_i$, anchored on a visited $i$-band cell (cf.~\cref{lem:liri-out}). 
Stage 2 attaches $q$ to $i$-beam($q$). 
\item If $S$ is parallel to $L_i$, then $(i-1)$-\beam($q$) is either empty or in the bottom partition of $\O_i$, anchored on a visited band cell. Stage 3 attaches $q$ to the adjacent $(i-1)$-band cell or $(i-1)$-\beam. 
\end{itemize}
This concludes the proof. 
\end{proof}

\paragraph{Unfolding the bottom of $\O_i$ when the $i$-band is not nearly-visited} In this case, $\beta$ is parallel to $(i-1)$-\beam($R_i$), which matches the orientations of the beams in the top partition of $\O_i$. Recall that the $(i-1)$-bridge is composed of beams parallel to $(i-1)$-\beam($L_i$). Thus, if $L_i$ is orthogonal to $R_i$, $\beta$ may cross the $(i-1)$-bridge. 

Assume first that $\beta$ does not cross the $(i-1)$-bridge. At this point, any unvisited $i$-band cell $r$ parallel to $R_i$ is positioned above $i$-\beam($r$) in the unfolding net (see~\cref{sec:top}). 
We attach $\beta$ to the unfolding net as follows:
\begin{itemize}
\squeezelist
\item If $\beta$ is anchored on an unvisited $i$-band cell $r$, then $r$ is parallel to $R_i$ and lies above $i$-\beam($r$) in the unfolding net. In this case, we attach $\beta$ above $r$, along their shared side. 
For example, see the bottom beams labeled $\beta_{41}$, $\beta_{42}$, $\beta_{51}$, $\beta_{52}$ and $\beta_{61}$ in~\cref{fig:ortho-ex-complete}.

\item If $\beta$ is anchored on a visited $i$-band cell $r$, and its second anchor is not on the $i$-band, we attach $\beta$ below $r$. See beam $\beta_{31}$ in~\cref{fig:ortho-ex-complete}, which is attached to its sole $3$-band anchor $w$ (above $w$, since $3$-pointer = $cw$). 

\item If both of $\beta$'s anchors are visited $i$-band cells, we attach $\beta$ below the \emph{second} visited $i$-band anchor $r$. This choice will be motivated further in~\cref{sec:allbands}. 
All unlabeled bottom beams in~\cref{fig:ortho-ex-complete} are examples of this case. 
\end{itemize}

\setlength{\intextsep}{0pt}%
\setlength{\columnsep}{1em}%
\begin{wrapfigure}[10]{r}{0.42\textwidth} 
\includegraphics[page=5,width=0.42\textwidth]{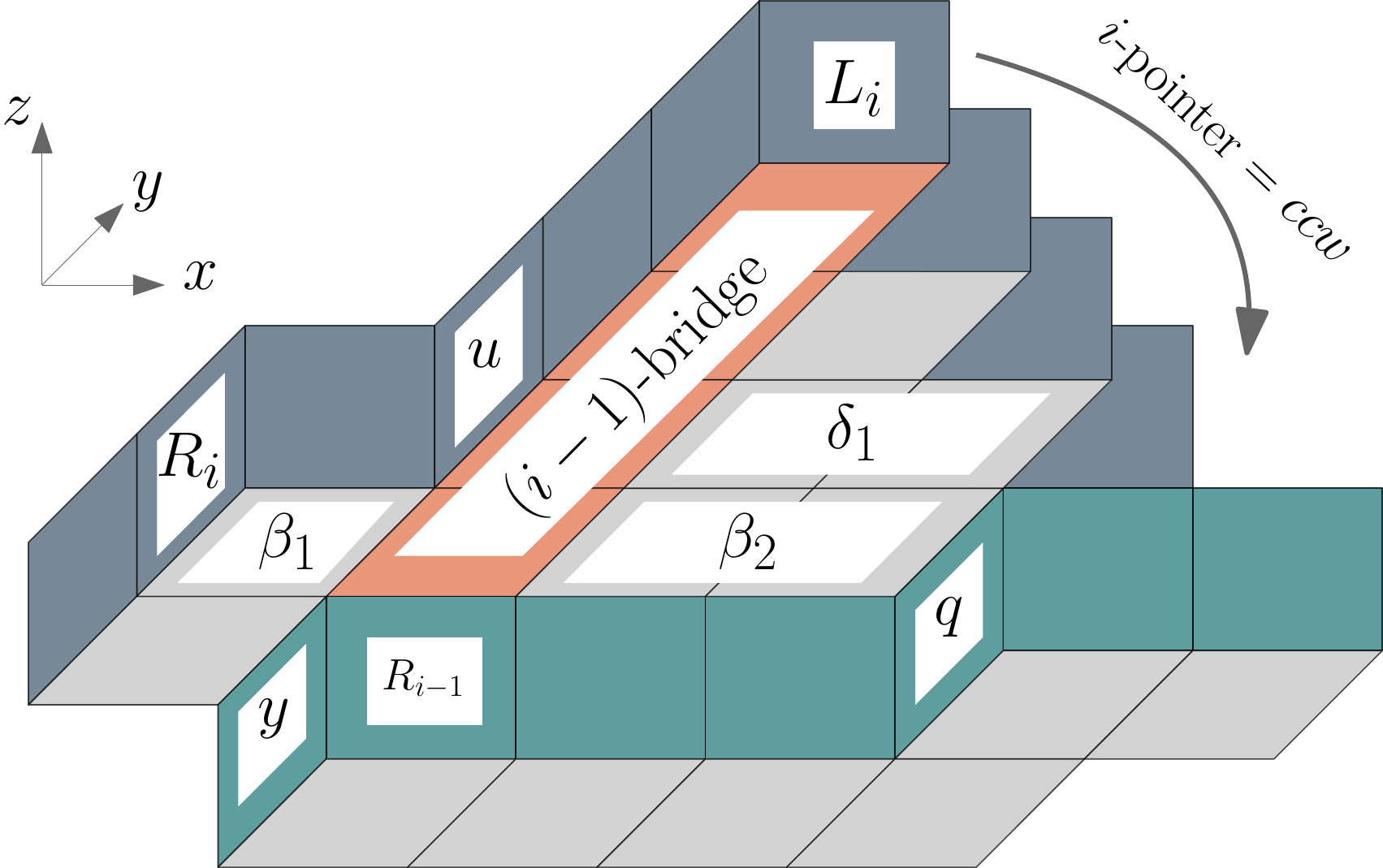}
\vspace{-2em}
    \caption{$\beta = (i-1)$-\beam($R_i$) crosses the $(i-1)$-bridge.} 
     \label{fig:ortho-bottom}
\end{wrapfigure}

Assume now that $\beta$ crosses the $(i-1)$-bridge. Refer to~\cref{fig:ortho-bottom}. By~\cref{lem:bridge-bot}, the $(i-1)$-bridge consists of a single strip (since it lies on a bottom face), and its intersection with $\beta$ is a single square piece. 
Let $\beta_1$ and $\beta_2$ be the two sub-\beams\ obtained by removing from $\beta$ the shared $(i-1)$-bridge section.
If both sub-\beams\ are empty, then  $\beta$ is entirely within the $(i-1)$-bridge and is already part of the unfolding net.  Otherwise, consider any non-empty sub-beam $\gamma \in \{\beta_1, \beta_2\}$: 
\begin{itemize}
\squeezelist
\item If $\gamma$ is anchored on an $i$-band cell $r$, then $r$ must be 
parallel to $R_i$ and thus already included in the unfolding net (see~\cref{sec:top}). 
In this case, we attach $\gamma$ to $r$: above if $r$ is unvisited, and below if $r$ is visited. 
For example, in~\cref{fig:ortho-bottom}, $\beta_1$ would be attached below $R_i$. 

\item Otherwise, $\gamma$ must be anchored on an $(i-1)$-band 
cell $q$. Since $\beta$ crosses the $(i-1)$-bridge, it is orthogonal to $(i-1)$-beam($R_{i-1}$), 
so $q$ is orthogonal to $R_{i-1}$. By~\cref{lem:liri-out}, $q$ is visited. 
In this case, we attach $\gamma$ to $q$ along their shared side (as is the case for $\beta_2$ in~\cref{fig:ortho-bottom}). 
\end{itemize}
Note that the entire bottom surface of $\O_i$ is now part of the unfolding net.

\subsection{Stage 4: Unfolding the Remaining Band Pieces of $\O$}  
\label{sec:allbands}%
At this stage, the unfolding net includes the top and bottom surfaces of $\O$, as well as all visited band cells and all $i$-band cells parallel to $R_i$, for each $i$. If the $i$-band is nearly-visited, the net covers the entire surface of $\O_i$ (cf.~\cref{lem:nearly-visited}).

For the cases where if the $i$-band is not nearly-visited, we complete the unfolding of $\O_i$ incrementally, for increasing $i$. 
Since $R_1$ and $L_1$ are adjacent, the entire surface of $\O_1$ is part of the unfolding net (as discussed in~\cref{sec:bands-bridges}).  

Fix $i > 1$ and assume that the entire surface of $\O_1, \ldots, \O_{i-1}$ has been incorporated in the unfolding net. 
Assume, without loss of generality, that $\O$ is oriented so that $R_i$ is a \emph{front} cell (with normal $+y$) and the $i$-$pointer$ is $ccw$ (the $cw$ case is symmetric). 
By~\cref{lem:front-allback}, all $i$-band back cells are visited, so the 
only surface pieces of $\O_i$ left to unfold are any unvisited left and right $i$-band cells. 
If there are no such cells,  the unfolding process for $\O_i$ is complete. 

Now, consider a \emph{straight} band segment $S = i$-band($a, b$) that contains unvisited left or right cells, 
where both $a$ and $b$ are parallel to $R_i$.
Let $S^* \subseteq S$ denote the unvisited portion of $S$. 
To unfold $S^*$, we consider each of the four exhaustive cases listed in~\cref{tab:bandunf}. 

\vspace{2em}
\begin{table}[htp]
\centering
\begin{tabular}{|c|c|c|}
\hline
\multicolumn{3}{|c|}{~~~~~{\bf Unfolding straight band segment $S^* \subseteq S = i$-band($a, b$)}~~~~~ }\\ 
\hline
{\bf $a$} & $b$ & unfolding procedure \\
\hline
\hline
unvisited & unvisited & \cref{sec:case1} \\
\hline
visited & unvisited & \cref{sec:case2} \\
\hline
visited & visited & \cref{sec:case3} \\
\hline
unvisited & visited & \cref{sec:case4} \\
\hline
\end{tabular}
\caption{Exhaustive case analysis for unfolding $S^*$.}
\label{tab:bandunf}
\end{table}

\vspace{1em}
\noindent
Let $a'$ and $b'$ denote the $i$-band cells opposite to $a$ and $b$, respectively. 

\subsubsection{Unfolding case 1: $a$ unvisited, $b$ unvisited.} 
\label{sec:case1}
In this case, $S^* = S$. 
By~\cref{lem:front-allback}, both $a$ and $b$ are front cells (since $R_i$ is assumed to be a front cell), and their opposite back cells $a'$ and $b'$ are visited. 

Let $S' = i$-band($b', a'$). Refer to~\cref{fig:unfops-1}. Note that $S'$ is a straight band segment orthogonal to $R_i$. Since both $a'$ and $b'$ are visited, $S'$ is entirely visited. Since $a$ is unvisited, $i$-\beam($a$) $\equiv i$-\beam($a'$) lies on top of $O_i$ (cf.~\cref{lem:liri-out}), and similarly for $b$ and $b'$. This, along with the $2$-manifold property of $\O$'s surface, implies that $i$-\beam($S'$) is non-empty and lies on the top of $\O_i$. 

\medskip
\begin{figure}[htp]
    \centering
    \includegraphics[page=1,width=0.8\textwidth]{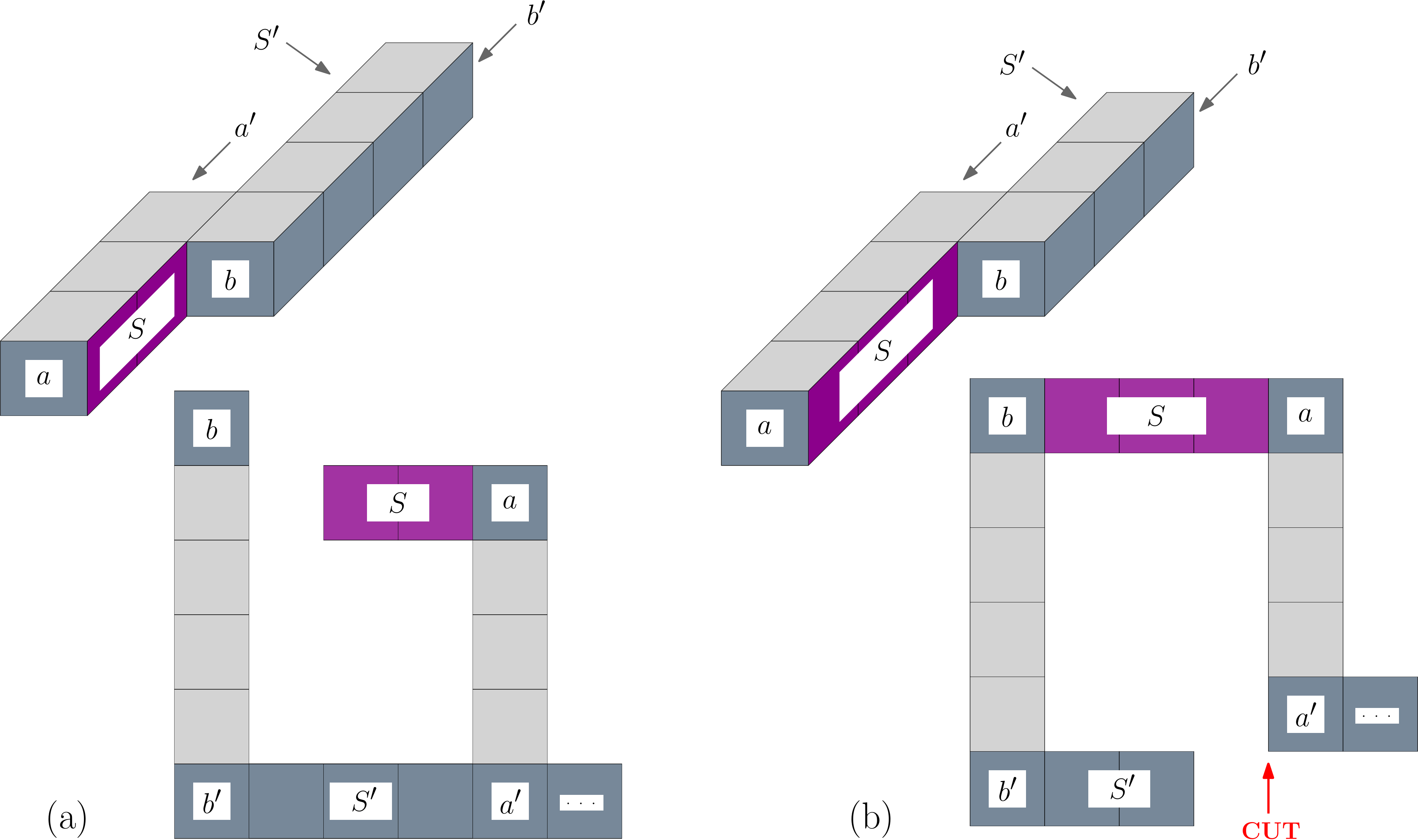} 
    \caption{Unfolding case 1: (a) $S$ no longer than $S'$ (b) $S$ longer than $S'$. }
    \label{fig:unfops-1} 
\end{figure}

\paragraph{Unfolding process}  
At this point, $a$ and $b$ are already in the unfolding net above $i$-\beam($a'$) and $i$-\beam($b'$), respectively (see~\cref{sec:top}). 
If $S$ is no longer than $S'$, we attach $S$ to $a$ (alternatively, to $b$). See~\cref{fig:unfops-1}(a). 
If $S$ is longer, we first cut the unfolding net along the left side of $a'$. (The cut could be along any vertical edge interior to $S'$, but for definiteness, we opt for the left edge of $a'$.) This cut disconnects the net into two components, which we then reconnect by attaching $S$ to both $a$ and $b$, so that $S$ unrolls horizontally between them. See~\cref{fig:unfops-1}(b).

\subsubsection{Unfolding case 2: $a$ visited, $b$ unvisited.} 
\label{sec:case2}
Given that $a$ (parallel to $R_i$) is visited and $b$ is unvisited, it follows that $a = R_i$ is a front cell. As $b$ remains unvisited, it must also be a front cell, because all back cells are visited (cf.~\cref{lem:front-allback}). Therefore, $S$ does not include $L_i$ and consequently $S^* = S$ (that is, the entire segment $S$ is unvisited). 

Since $b$ is unvisited, $i$-\clip($b$) lies on top of $\O_i$ and is not adjacent to $L_{i+1}$ (cf.~\cref{lem:liri-out}). %
Similarly, since $S$ is unvisited, $i$-\beam($S$) is non-empty and lies on top of $\O_i$. These together imply that the $i$-bridge lies on top of $\O_i$ and consists of a single beam, namely $i$-\beam($a$). This implies that the $(i+1)$-pointer is identical to the $i$-pointer ($ccw$, by our assumption). 

\vspace{1em}
\begin{figure}[htp]
    \centering
    \includegraphics[page=1,width=0.96\textwidth]{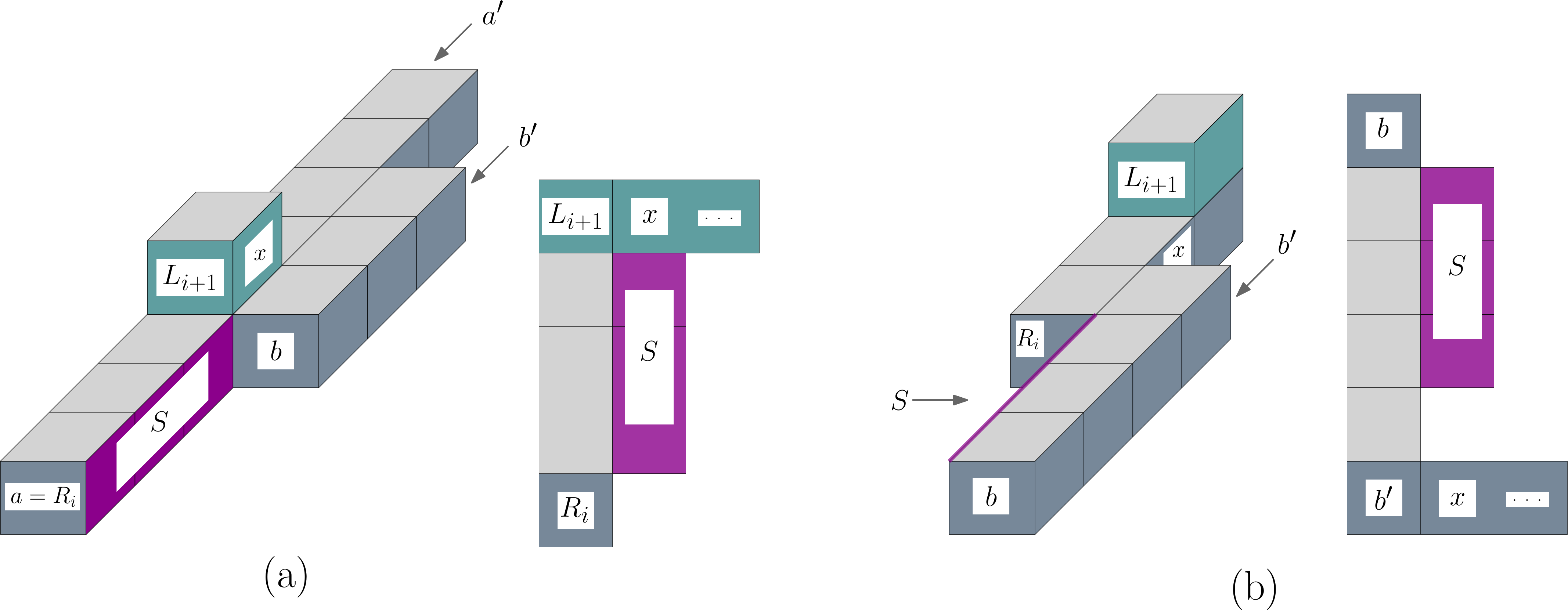} 
    \caption{Unfolding case 2: (a) $S$ facing \emph{right}  (b) $S$ facing \emph{left}.}
    \label{fig:unfops-1copy} 
\end{figure}

\paragraph{Unfolding process} If $S$ is a right segment, we attach $S$ to the right side of $i$-\beam($a$) in the unfolding net, along their shared side. See~\cref{fig:unfops-1copy}(a).
If $S$ is a left segment, we attach $S$ to the right 
side of $i$-\beam($b'$) in the unfolding net, along their shared side. See~\cref{fig:unfops-1copy}(b).

\subsubsection{Common scenario.} 
\label{sec:common}
Before discussing the next two cases, we highlight a scenario common to both. 
In this setting, $L_i$ is adjacent to an unvisited $i$-band cell $r$, the interior angle between $L_i$ and $r$ is $\pi/2$, and $S = S^*$ is the maximal straight $i$-band segment $S$ containing $r$. See~\cref{fig:commoncopy}(a). 
The process of unfolding $S$ requires a complex analysis, deferred to~\ref{sec:appendix-common} to maintain focus in this section.

\subsubsection{Unfolding case 3: $a$ visited, $b$ visited.} 
\label{sec:case3}
Given that $a$ (which is parallel to $R_i$) is visited and the segment $S$ (orthogonal to $R_i$) is not entirely visited, it follows that $a = R_i$ is a front cell. Since $b$ is visited, either $b = L_i$ or $S$ contains $L_i$. If $L_i \in S$, then the $i$-band is nearly-visited, in which case $S$ is already part of the unfolding net (cf.~\cref{lem:nearly-visited}.) 
Therefore, the discussion in this section assumes that $b = L_i$, so $S^* = S$. 
By~\cref{lem:front-allback}, $L_i$ cannot be opposite to $R_i$, therefore $L_i$ is also a front cell. 

\vspace{1em}
\begin{figure}[htp]
    \centering
    \includegraphics[page=3,width=0.88\textwidth]{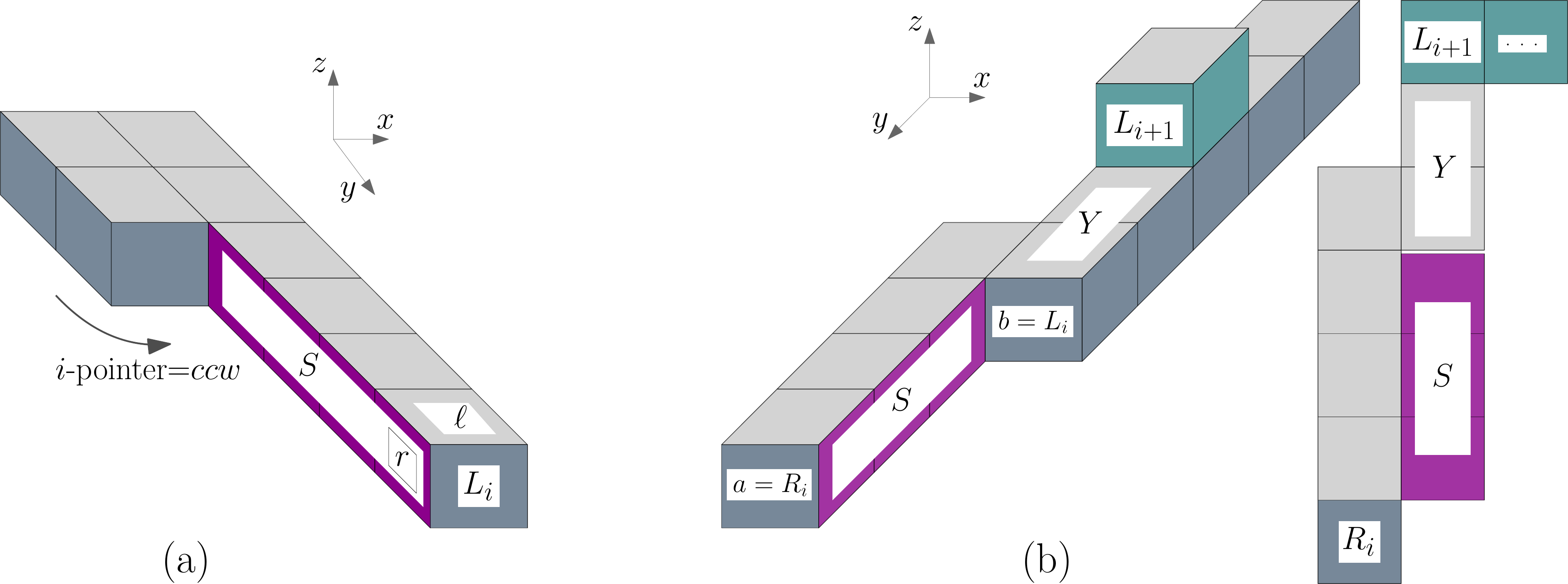}
    \caption{(a) Common scenario. (b) Unfolding case 3, $S$ facing \emph{right}.}
    \label{fig:commoncopy}
\end{figure}

\paragraph{Unfolding process} 
The case where $S$ is a left segment matches the setting for the common scenario and its unfolding is detailed in~\ref{sec:appendix-common}. 
If $S$ is a right segment, we attach it to the right side of $i$-\beam($R_i$). See~\cref{fig:commoncopy}(b). 

\subsubsection{Unfolding case 4: $a$ unvisited, $b$ visited.} 
\label{sec:case4}
Given that $a$ is unvisited and $R_i$ is a front cell, $a$ must also be a front cell, because all back calls are visited (cf.~\cref{lem:front-allback}).
Moreover, $i$-\beam($a$) lies on top of $O_i$  and is anchored on the $i$-band only (cf.~\cref{lem:liri-out}). 
Since $b$ is visited, either $b = L_i$ (in which case $S^* = S$) or $S$ includes $L_i$.  
 


\begin{figure}[htp]
    \centering
    \includegraphics[page=2,width=0.94\textwidth]{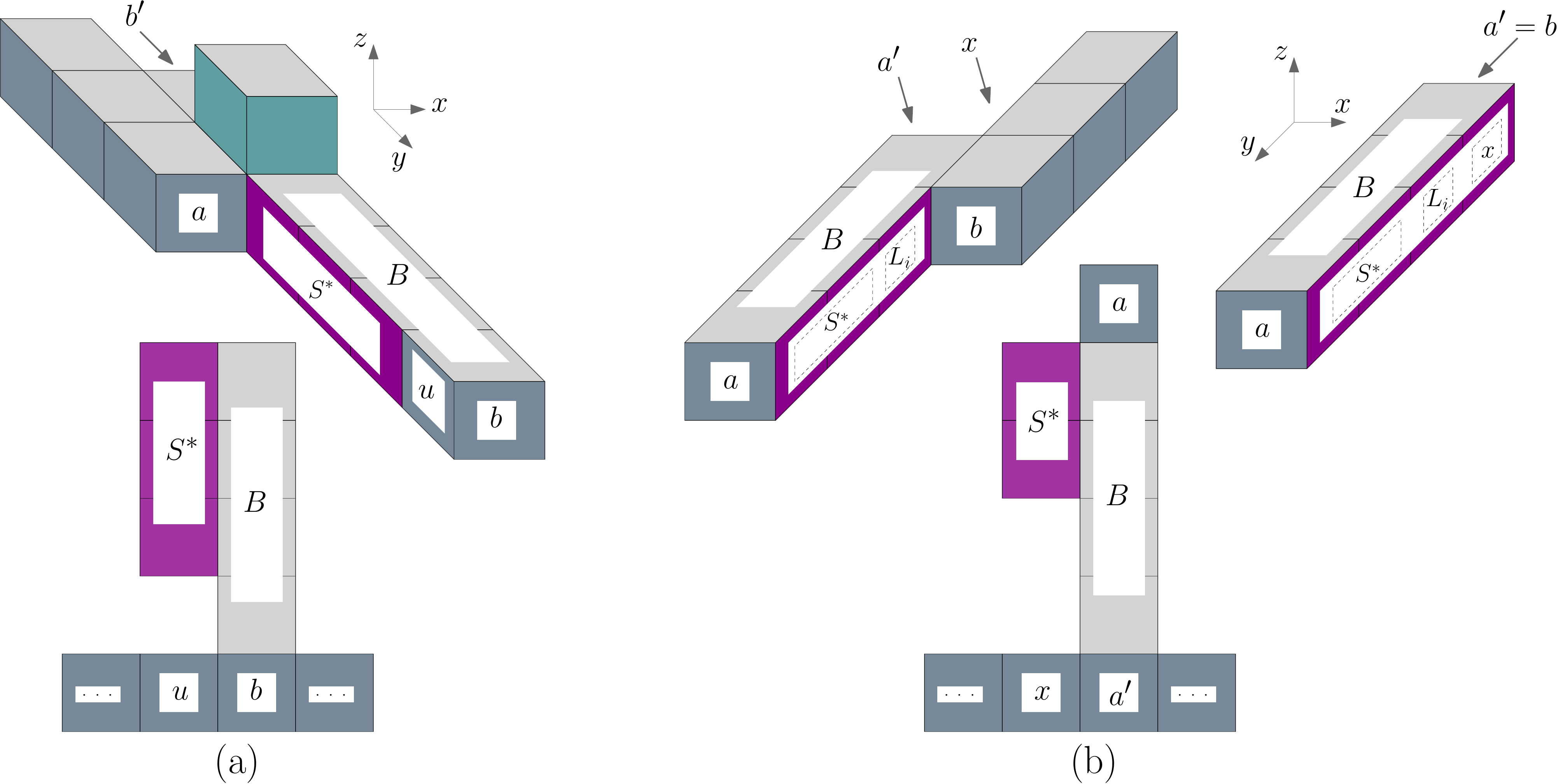}
    \caption{Unfolding case 4: (a) $S$ facing \emph{left} (b) $S$ facing \emph{right}.}
    \label{fig:unfopscopy-4}
\end{figure}

\paragraph{Unfolding process} 
Our unfolding procedure depends on whether $S$ is a left or right segment. 
Assume first that $S$ a left segment. Then $b$ must be a front cell. 
If  $b = L_i$, this matches the common scenario detailed in~\ref{sec:appendix-common}.
If $b \neq L_i$, then $S$ includes $L_i$, so $S^* \subset S$. 
We willl show that $S^*$ is adjacent to $B = i$-\beam($b$). 
We first relocate $B$ to sit on top of $b$ in the unfolding net (if not already there), 
then attach $S^*$ to the left side of $B$. See~\cref{fig:unfopscopy-4}(a). 

Assume now that $S$ a right segment. In this case, $b$ may be a front cell, or a back cell identical to $a'$. If $b$ is a front cell, or a back cell distinct from $L_i$, we attach $S^*$ to the left side of the $i$-\beam($a'$), as shown in~\cref{fig:unfopscopy-4}(a). 
If $b =  L_i$ is a back cell, this again matches the common scenario detailed in~\ref{sec:appendix-common}.

\medskip
\noindent
Having exhausted all cases, the unfolding process is now complete.

\subsection{Complete unfolding example}
\cref{fig:ortho-ex-complete} shows the final unfolding net for the polycube example from~\cref{fig:defs-1}(a). All visited band segments extend horizontally, and top and bottom beams extend vertically above or below adjacent  band cells. 

\vspace{1em}
\begin{figure}[h!]
    \centering
    \includegraphics[page=1,width=\textwidth]{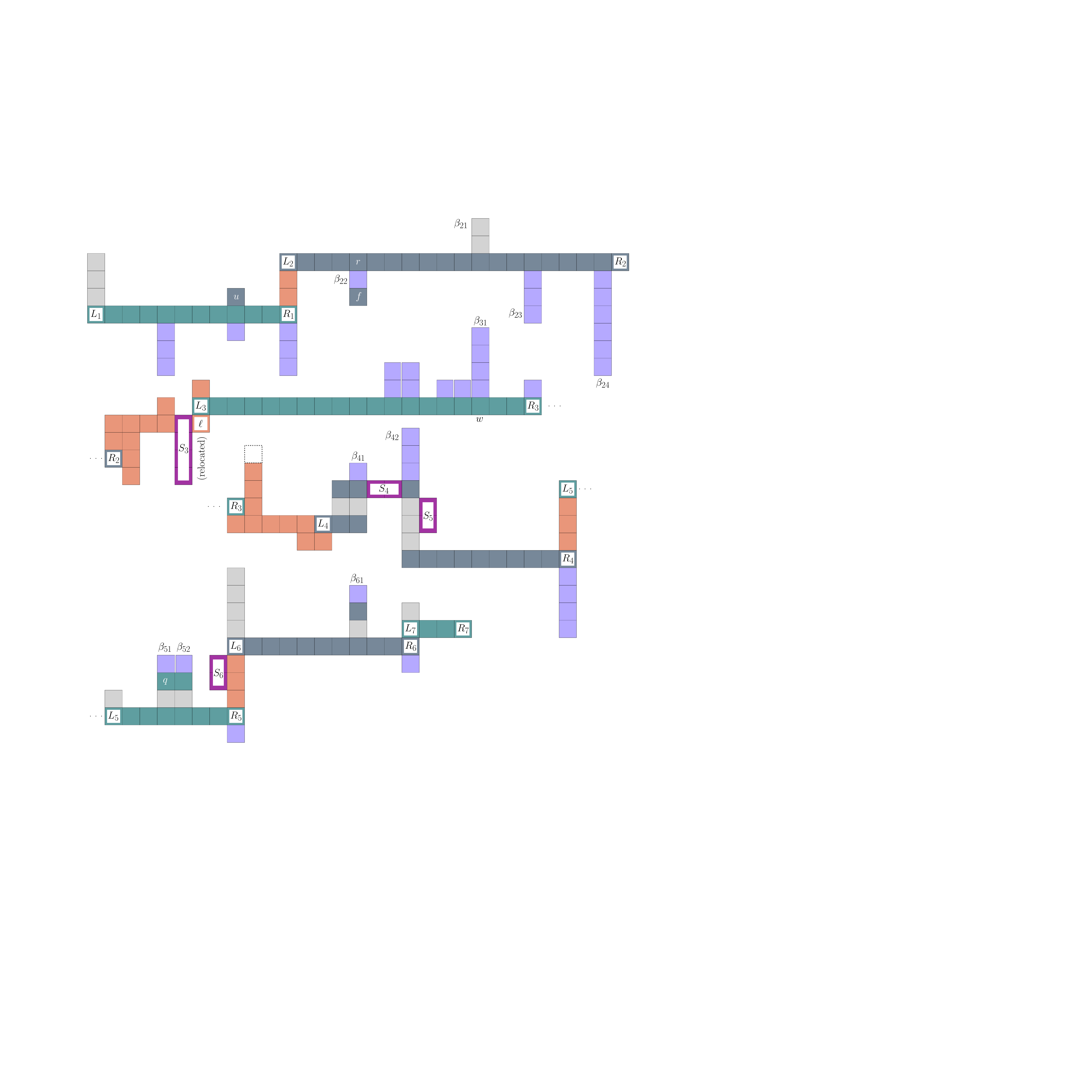}
    \caption{Complete unfolding of the polycube from~\cref{fig:defs-1}(a): horizontal visited band segments are connected by orange bridges; top beams are marked in light gray, and above some of them are unvisited band cells; bottom beams are marked in light purple; and $S$-labeled segments represent unvisited left or right band segments.}
    \label{fig:ortho-ex-complete}
\end{figure}

\noindent
Observe the following:
\begin{itemize}
\item The entire 1-band is visited, and the $1$-pointer is $ccw$.  The top and bottom of $\O_1$ are partitioned into beams parallel to $1$-\beam($R_1$), which are attached above and below the 1-band, respectively. 

\item The $2$-band is nearly-visited. There is a single top beam $\beta_{21}$ parallel to $2$-\beam($R_2$), which attaches above the $2$-band. The bottom of $\O_2$ is partitioned into beams parallel to $1$-\beam($L_2$), each attaching to an adjacent $2$-band cell, with ties broken in favor of the most recently visited cell (see $\beta_{22}$, $\beta_{23}$, and  $\beta_{24}$). 
The two unvisited $2$-band faces, $u$ and $f$, attach to the adjacent $1$-band cell and bottom $1$-beam, respectively. 

\item $L_3$ and $R_3$ are separated by a straight, unvisited $3$-band segment labeled $S_3$, and the $3$-pointer is $cw$. The 3-bridge covers the entire top surface of $\O_3$. The bottom of $\O_3$ is partitioned into beams parallel to $3$-\beam($R_3$), which are positioned above the 3-band (since the $3$-pointer is $cw$). 

From a point of view facing $R_3$, the configuration for $S_3$ is a horizontal mirror image of case 3 discussed in~\cref{sec:case3}. Note that the top cell $\ell$ adjacent to $L_3$, originally part of $3$-bridge, 
has been relocated below $L_3$ to facilitate the unfolding of $S_3$. 

\item For any $i > 3$, the $i$-band is not nearly-visited. Both the top and bottom of $\O_i$ are partitioned into beams parallel to $i$-\beam($R_i$). The top beams are positioned above the $i$-band, with any adjacent unvisited front $i$-band cells attached above them, and any adjacent bottom $(i-1)$-beams attached above those cells (see $\beta_{41}$, $\beta_{42}$, $\beta_{51}$, $\beta_{52}$, and $\beta_{61}$). 

Each bottom $(i-1)$-beam not adjacent to an unvisited $i$-band cell is attached to a visited $i$-band cell, with ties broken in favor of the most recently visited cell (see the beams attached below $R_4$, $R_5$ and $R_6$). 

\item From a point of view facing $R_4$, $S_4$ is a right segment matching the 
configuration of case 1 discussed in~\cref{sec:case1}, and $S_5$ is a left segment matching the configuration of case 2 discussed in~\cref{sec:case2}.
\item From a point of view facing $R_6$, $S_6$ is a left segment matching the configuration of case 4 discussed in~\cref{sec:case4}.
\end{itemize}
This example illustrates most of the subtle scenarios of our unfolding approach. 

\section{Conclusion}

We show that every polycube with orthogonally convex layers can be edge unfolded. 
A natural extension of this work would be to remove the orthogonal convexity constraint, thereby  
addressing the general polycube edge unfolding problem posed in~\cite{open-problems}. 

Another potential extension would be to use rectangular boxes instead of cubes as building blocks for the 
orthogonal polyhedron. Our unfolding algorithm works as-is for rectangular boxes, provided their height does 
not exceed their width and depth. 
Removing this height restriction would move us closer to the goal of unfolding arbitrary orthostacks, 
a challenge that remains open. 

\bibliographystyle{splncs04}
\bibliography{unfolding}

\newpage
\appendix
\input{appendix}

\input{proof}

\end{document}

%% file: appendix.tex
\noindent
\section{Structural Properties of Polycubes}
\label{sec:appendix-lemmas}

This section provides proofs for the lemmas stated in~\cref{sec:bands-bridges}, along with two new lemmas needed to establish the correctness of the unfolding approach.

\liriexist*
\begin{proof} 
We first show that $R_i$ and $L_{i+1}$ exist. If the $i$-band contains a cell $a$ adjacent to a cell $b$ in the $(i+1)$-band, then $i$-\clip($a$) is empty, making $a$ and $b$ valid candidates for $R_i$ and $L_{i+1}$, respectively (cf.~\cref{def:ri}). 
Otherwise, by~\cref{lem:opposite-edges}., there exists an $i$-beam anchored on some $a \in i$-band and $b \in (i+1)$-band.  
By~\cref{def:iclip}, $i$-\clip$(a)$ is adjacent to $b$, so $a$ and $b$ again serve as valid candidates for $R_i$ and $L_{i+1}$. 
Since $a$ lies on the walk around the $i$-band, both $R_i$ and $L_{i+1}$ exist. 

We now show that they are uniquely defined. By definition, $R_i$ is the last $i$-band cell (in the $i$-pointer direction) satisfying the conditions of~\cref{def:ri}, so it is unique. 
If $i$-\clip$(R_i)$ is empty, the adjacent $(i+1)$-band cell uniquely determines $L_{i+1}$.

If $i$-\clip$(R_i)$ is non-empty, we partition it into beams parallel to $i$-\beam($R_i$).
From $i$-\beam($R_i$), walk in the direction of the $i$-pointer, passing over all beams anchored solely on the $i$-band. 
Let $\beta$ be the {\bf last such beam} encountered. Immediately after $\beta$, by the definition of $R_i$, there must be a beam anchored on an $(i+1)$-band cell. 
By the orthogonal convexity of $\O_{i+1}$, there are at most two such beams. 
If only one exists, its $(i+1)$-band anchor uniquely defines $L_{i+1}$.  
If two such beams exist, we select the one whose $(i+1)$-band anchor is closer to $R_i$ in Manhattan distance, ensuring uniqueness. 
\end{proof}

\beamends*
\begin{proof}
Let $B$ be any beam on the surface of $\O_i$, anchored on band cells $a$ and $b$. 
If $B$ is a top $i$-\beam\ and both $a$ and $b$ are on the $(i+1)$-band, then an infinite extension of $B$ in a direction orthogonal to $a$ would intersect $\O_{i+1}$ into two disconnected components (separated by $B$), contradicting its orthogonal convexity.
Therefore, at least one of $a$ and $b$ must be on the $i$-band. 
Similar reasoning applies if $B$ is a bottom $(i-1)$-\beam. 
\end{proof}

\liriout*
\begin{proof}
We show that, if $u$ does not meet these conditions, it qualifies as a potential candidate for $R_i$. 
If $i$-\clip($u$) is empty, then $u$ is a valid candidate for $R_i$ (cf.~\cref{def:ri}). 
If $i$-\clip($u$) lies on the bottom of $\O_{i+1}$, so does $i$-\beam$(u)$, which is a subset of $i$-\clip$(u)$. By~\cref{lem:beam-ends}, $i$-\beam($u$) is anchored on an $(i+1)$-band cell $b$ parallel to $u$, making $u$ a valid candidate for $R_i$. The same holds if the final condition of the lemma is not met. 

Thus, in all cases, $u$ is a valid candidate for $R_i$. But then $u$ must be visited --- otherwise, the walk around the $i$-band would have continued past $R_i$ to reach $u$, contradicting the choice of $R_i$ from~\cref{def:ri}. This contradicts the lemma's assumption that $u$ is unvisited. 
\end{proof}

\bridgebot*
\begin{proof}
The $i$-bridge includes $i$-\beam($R_i)$. If $i$-\beam($R_i)$  is on the bottom of $\O_{i+1}$, then by~\cref{lem:beam-ends}, it is adjacent to an $(i+1)$-band cell parallel to $R_i$. By~\cref{def:ri}, this cell is precisely $L_{i+1}$, so $i$-\beam($R_i)$ is the bridge.
\end{proof}

\medskip
\noindent
The following lemma, whose proof is nontrivial, plays a key role in proving the correctness of the unfolding algorithm. 

\lemstrips*
\begin{proof}
By~\cref{lem:liri-out}, $i$-\beam$(u)$ is non-empty and lies on top of $\O_{i}$ (otherwise $u$ is visited). We prove this lemma by contradiction. Let $w \neq u$ be the other anchor of $i$-\beam$(u)$,  and assume,  
for the sake of contradiction, that $w$ is also unvisited. 
Note that $u$, $w$ and $R_i$ are parallel. 
Since $i$-\clip$(u)$ includes $i$-\beam$(w) \equiv i$-\beam($u$), it is adjacent to $w$. By~\cref{lem:liri-out}, $w$ is not on the $(i+1)$-band, therefore both $u$ and $w$ are on the $i$-band. 

Let $F$ be the $i$-face containing $i$-\beam($u$). By~\cref{lem:opposite-edges}, the boundary of $F$ is partitioned into two contiguous components: one component $\partial F_i$ adjacent to the $i$-band, and one component $\partial F_{i+1}$ adjacent to the $(i+1)$-band. Swap the labels $u$ and $w$, if necessary, so that a walk around the $i$-band in the direction of the $i$-pointer, starting at $u$, encounters $w$ before reaching $\partial F_{i+1}$. Refer to~\cref{fig:longband}(a). 
We prove the following property:

\medskip
\noindent
(a) $\partial F_{i+1}$ is a straight line segment orthogonal to $R_i$. 

\medskip
\noindent
Assume to the contrary that $\partial F_{i+1}$ is adjacent to an $(i+1)$-band cell $b$ parallel to $R_i$. By~\cref{lem:beam-ends}, $i$-\beam($b$) is anchored on an $i$-band cell $a$. Recall that $u$ and $w$ are positioned such that a walk in the direction of the $i$-pointer, starting at $w$, 
encounters $a$ before $u$. 
By definition, 
$i$-\clip($w$) includes $i$-\beam($a$) and is, therefore, adjacent to $b$. This, along with~\cref{lem:liri-out}, contradicts 
 our assumption that $w$ is unvisited. This settles property (a). 

\begin{center}
\begin{figure}[htp]
    \includegraphics[page=1,width=0.95\textwidth]{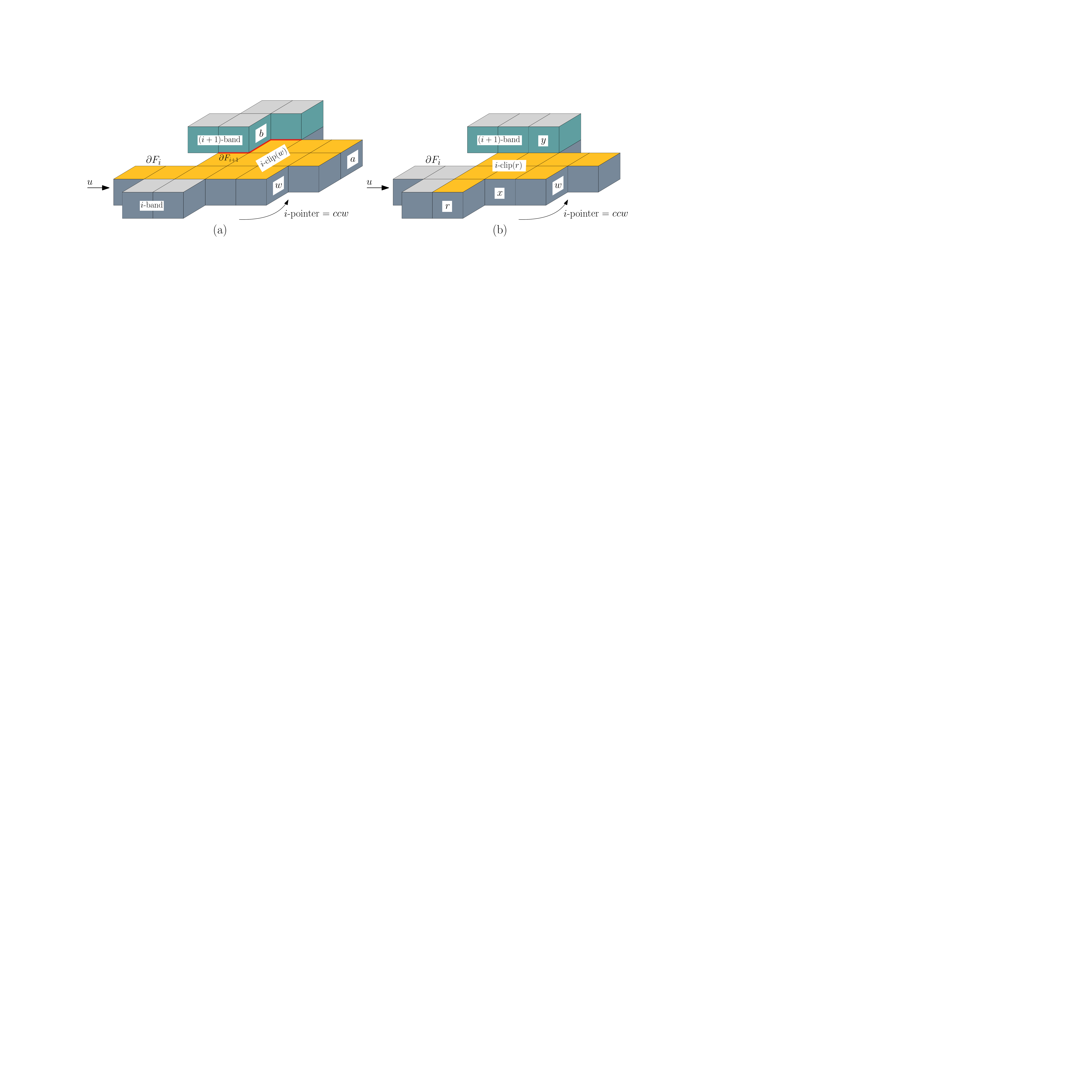} 
    \caption{\cref{lem:strips}: (a) $i$-\clip($w$) is adjacent to (i+1)-band cell $b$ parallel to $w$ (b) $i$-\clip($r$) is adjacent to $(i+1)$-band cell $y$ parallel to $r$.}
    \label{fig:longband}
\end{figure}
\end{center}

By property (a) above, $F$ is not adjacent to $L_{i+1}$ (which, by definition, is parallel to $R_i$). Consequently, $F$ is 
also not adjacent to $R_i$, since $R_i$ and $L_{i+1}$ are adjacent to the same $i$-face. 
Therefore, a walk around the $i$-band in direction of the $i$-pointer, starting at $R_i$, encounters $u$ before $w$. 
As both $u$ and $w$ are unvisited, $i$-band$[u, w]$ cannot contain $L_i$, so the entire $i$-band$[u, w]$ is unvisited. 
Furthermore, since $u$ and $w$ are parallel and orthogonal to $\partial F_{i+1}$, $i$-band$[u, w]$ must contain a cell $r$ parallel to $\partial F_{i+1}$. 

Let $y$ be any $(i+1)$-band cell adjacent to $F$, and let $x \neq y$ be the other anchor of $i$-\beam($y$). Refer to~\cref{fig:longband}(b). By~\cref{lem:beam-ends}, $x$ lies on the $i$-band. By~\cref{lem:liri-out}, $x$ is visited. 
Since $R_i$ is not adjacent to $F$, a walk in the direction of the $i$-pointer around the $i$-band, starting from $L_i$, encounters $x$, $R_i$ and $r$ in that order.  Since $r$ and $x$ are adjacent to $F$ and share the same normal, and $R_i$ is not adjacent to $F$, $x$ must be adjacent to $i$-\clip($r$). This implies that $i$-\clip($r$) includes $i$-\beam($x$), which is anchored on $y$, contradicting the fact that $r$ is unvisited. This concludes the proof. 
\end{proof}

\medskip
\noindent
Note that~\cref{lem:strips} does not always hold for $i$-band cells orthogonal to $R_i$. In~\cref{fig:longband2}(a) for example, opposite  cells $u$ and $w$ orthogonal to $R_i$ are unvisited. 

\vspace{1em}
\begin{figure}[htp]
\centerline{\includegraphics[page=2,width=\textwidth]{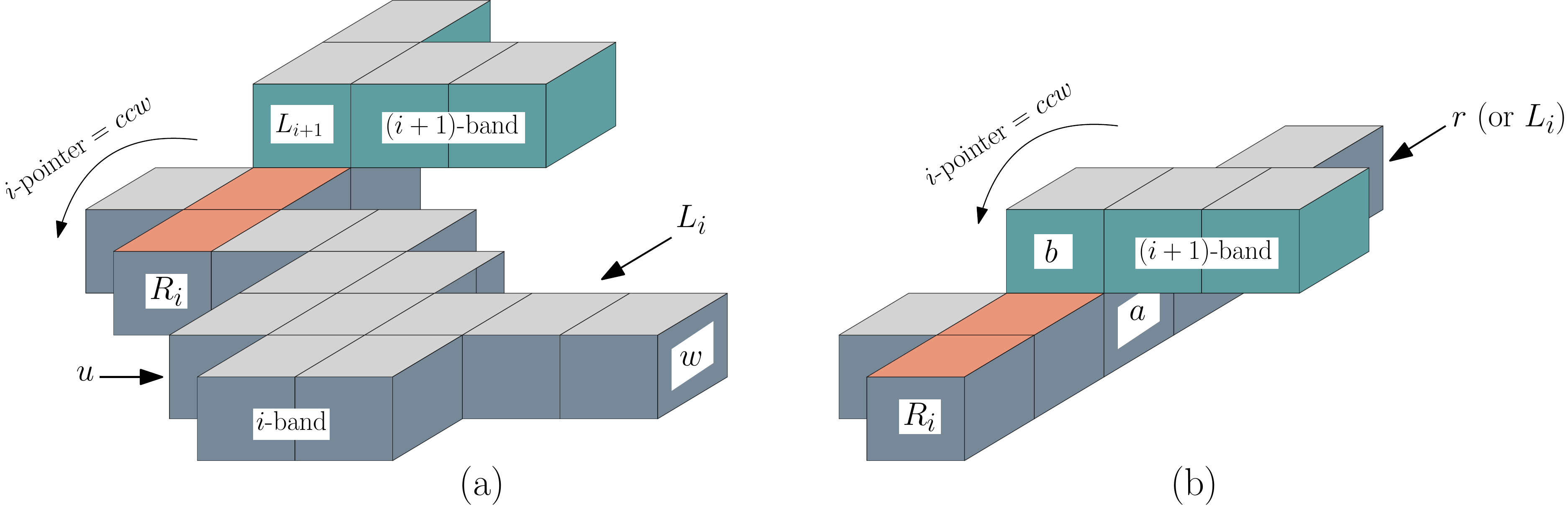}}
\caption{(a) \cref{lem:strips} does not hold for $i$-band cells orthogonal to $R_i$: opposite cells $u$ and $w$ are both unvisited. (b) \cref{lem:front-allback}: $a$ must be visited.}
\label{fig:longband2}
\end{figure}

\frontallback*
\begin{proof}
Assume, without loss of generality, that $R_i$ is a front cell and the $i$-\emph{pointer} is $ccw$; the other cases are symmetric. 
To prove the first part of the lemma, we show that all $i$-band back cells are visited. Suppose, for contradiction, that there is an unvisited $i$-band back cell. Since $R_i$ is a front cell and the traversal is $ccw$, the \emph{rightmost} back cell $r$ must be unvisited. By~\cref{lem:strips}, $i$-\beam($r$) is anchored on a visited $i$-band cell. It follows that $R_i$ must be the \emph{rightmost} front $i$-band cell, and thus $i$-\clip($R_i$) coincides with $i$-\beam($R_i$), which is not anchored on the $(i+1)$-band---contradicting the definition of $R_i$.


\medskip
We now prove the second part of this lemma. Suppose, for contradiction, that $L_i$ and $R_i$ are opposite to each other. 
We have already established that all $i$-band cells with normals opposite to that of $R_i$ are visited. 
Together, these imply that $R_i$ and $L_i$ are the rightmost front and back cells, respectively, 
and thus the straight $i$-band $S$ segment connecting $R_i$ to $L_i$ is entirely unvisited. 
By~\cref{def:ri}, either $R_i$ or $i$-\beam($R_i$) is anchored on an $(i+1)$-band cell $b$ parallel to $R_i$. Refer to~\cref{fig:longband2}(b). 
By the 2-manifold property of the surface of $\O$, the unvisited segment $S$ must include a cell $a$ that shares its top front vertex with $b$. 
Hence, $i$-\beam($a$) is either empty or located on the bottom of $\O_{i+1}$. This implies that either $b$ or $i$-\beam($b$) is anchored on an $(i+1)$-band cell. By~\cref{lem:liri-out}, cell $a$ must be visited. This contradicts the fact that $a$ is part of the unvisited segment $S$. 
\end{proof}

\begin{lemma}
The $i$-bridge is connected and orthogonally convex. 
\label{lem:bridge-convex}
\end{lemma}
\begin{proof}
If the $i$-bridge consists of a single beam $i$-\beam($R_i$), the lemma trivially holds. Suppose now that the $i$-bridge contains multiple beams. 
Except for $i$-\beam($L_{i+1}$), all beams in the bridge must be anchored exclusively on the $i$-band. Indeed, if any other  beam were anchored on an $(i+1)$-band cell $b \neq L_{i+1}$, then $b$ would be a better candidate for $L_{i+1}$---hav ing fewer beams between $i$-\beam($R_i$) and $i$-\beam($b$)---contradicting~\cref{def:ri}. 

This, combined with the orthogonal convex of $\O_i$, implies that the portion $\beta$ of the $i$-bridge formed by removing $i$-\beam($L_{i+1}$) is itself connected and orthogonally convex. Since $i$-\beam($L_{i+1}$) lies entirely within the boundary of $\O_{i}$ (as it shares an $i$-face with $\beta$), reattaching it to $\beta$ maintains both connectivity and orthogonal convexity. The lemma follows.
\end{proof}



\begin{lemma}
If the $i$-bridge consists of multiple beams, all $i$-band cells parallel to $R_i$ are visited. 
\label{lem:bridge-visited}
\end{lemma}
\begin{proof}
Assume, without loss of generality, that $R_i$ is a front cell and the $i$-pointer is $ccw$; the other cases are symmetric. In this scenario, $i$-\beam($R_i$) and $i$-\beam($L_{i+1}$) are the leftmost and rightmost $i$-bridge beams, respectively. By~\cref{def:ri}, $i$-\clip($R_i$) includes the entire $i$-bridge. Moreover, for any front cell $f$ anchored on the $i$-bridge, $i$-\clip($f$) is adjacent to $L_{i+1}$, as it includes the portion of the $i$-bridge bounded by $i$-\beam($f$) and $i$-\beam($L_{i+1}$). It follows that the front cell $f \neq R_i$ first encountered in a $ccw$ walk from $R_i$ must be visited. Consequently, all front cells are visited. By~\cref{lem:front-allback}, all back cells are visited. This completes the proof. 
\end{proof}


\newpage
\section{Common unfolding scenario} 
\label{sec:appendix-common}
The setting for the shared unfolding scenario is as follows: $L_i$ is adjacent to an unvisited $i$-band cell $r$, the interior angle formed by $L_i$ with $r$ is $\pi/2$, and $S = S^*$ is the maximal straight $i$-band segment $S$ containing $r$ (orthogonal to $L_i$). We also assume, without loss of generality, that $i$-pointer = $ccw$ (the case with $i$-pointer = $cw$ is symmetric), so $S$ is a \emph{left} segment, as shown in~\cref{fig:commoncopy}(a). We discuss the process of unfolding $S$. 

\medskip
\noindent
To simplify our discussion, reorient $\O$ (if necessary) such that $L_i$ is a front cell. 
Since $S$ is unvisited, $i$-\beam($S$) is non-empty, lies on top of $\O_i$, and is anchored on the $i$-band only (cf.~\cref{lem:liri-out}). 
If $R_i$ is parallel to $S$, then $S$ has already been added to the the unfolding net in Stage 2 of the unfolding algorithm (see~\cref{sec:top}). 
Therefore, the discussion in the section is based on the assumption that {\bf $R_i$ is parallel to $L_i$}, and consequently 
the $i$-band is not nearly-visited (by definition). 

The unfolding process for $S$ depends on whether the cell $\ell \in i$-\beam($L_i$) adjacent to $L_i$ is attached to $L_i$ in the unfolding net. 
In some cases, this cell is pivotal to attaching $S$ to the net. Therefore, we first attempt to relocate $\ell$ by attaching it to $L_i$ in the unfolding net. 

\paragraph{Relocating $\ell$} Let $L = i$-\beam($L_i$). 
If $L$ is not part of the $i$-bridge, it sits on top of $L_i$ in the unfolding net (see~\cref{sec:top}), so 
no relocation is necessary. If $L$ is part of the $i$-bridge, detaching $\ell$ from $L$ may potentially disconnect the $i$-bridge. 
One such example is depicted in~\cref{fig:unfops-common-5}, where the $i$-bridge covers the entire $i$-face. 
If removing $\ell$ from the $i$-bridge disconnects it, we leave  $\ell$ in place. 
Otherwise, we detach $\ell$ from $L$ and reattach it to $L_i$ in the unfolding net. 

\paragraph{Unfolding process} To determine suitable placements for $S$ that avoid overlap, we examine the structure of the $(i-1)$-bridge.
Let $B = (i-1)$-\beam($L_i$). By~\cref{def:ri}, $B$ is part of the $(i-1)$-bridge. 
Our analysis covers four distinct scenarios, summarized in Table~\ref{tab:common-scenarios}.

\vspace{1em}
\begin{table}[ht]
\centering{
\begin{tabular}{|c|c|c|c|}
\hline
{\bf Common} &  \multirow{2}{*}{{\bf Net cell $\ell$}} & \multirow{2}{*}{{\bf Beam $B$}} &  \multirow{2}{*}{{\bf $(i-1)$-bridge}} \\
{\bf Scenario} & ~ & ~ & ~ \\
\hline
\hline
1 &  \multirow{4}{*}{connected to $L_i$} & empty, or & single beam \\
\cline{1-1} \cline{4-4} 
2 & ~ & on top of $\O_{i-1}$ & multiple beams \\ 
\cline{1-1} \cline{3-4} 
 \multirow{2}{*}{3} & ~ & non-empty, and  &   \multirow{2}{*}{~} \\
~ & ~ & on the bottom of $\O_i$ &~\\ 
\hline
4 & disconnected from $L_i$ & \multicolumn{2}{c|}{~} \\
\hline
 \end{tabular}
}
\caption{Common unfolding scenarios.} 
\label{tab:common-scenarios}
\end{table}

\vspace{1em}
To simplify our discussion, we introduce some notation. 
Let $f$ be the front  $(i-1)$-band cell adjacent to $B$ (or to $L_i$, if $B$ is empty). 
Let $y$ be the $(i-1)$-band cell $cw$ from $f$, and $Y = (i-1)$-\beam($y$). Let $y' \neq y$ be the other anchor of $Y$. 

 
\paragraph{Common scenario 1} In this scenario, $\ell$ sits on top of $L_i$ in the unfolding net (after relocation, if necessary); 
$B$ is empty or on the top of $\O_{i-1}$; and the $(i-1)$-bridge is either empty (if $B$ is empty) or consists of a single beam, namely $B$.  
Consequently, $f = R_{i-1}$, $(i-1)$-pointer = $ccw$, and $y$ is attached to the left of $R_{i-1}$ in the unfolding net. 

\vspace{1.5em}
  \begin{figure}[hpt]
    \centering
    \includegraphics[page=2,width=0.98\textwidth]{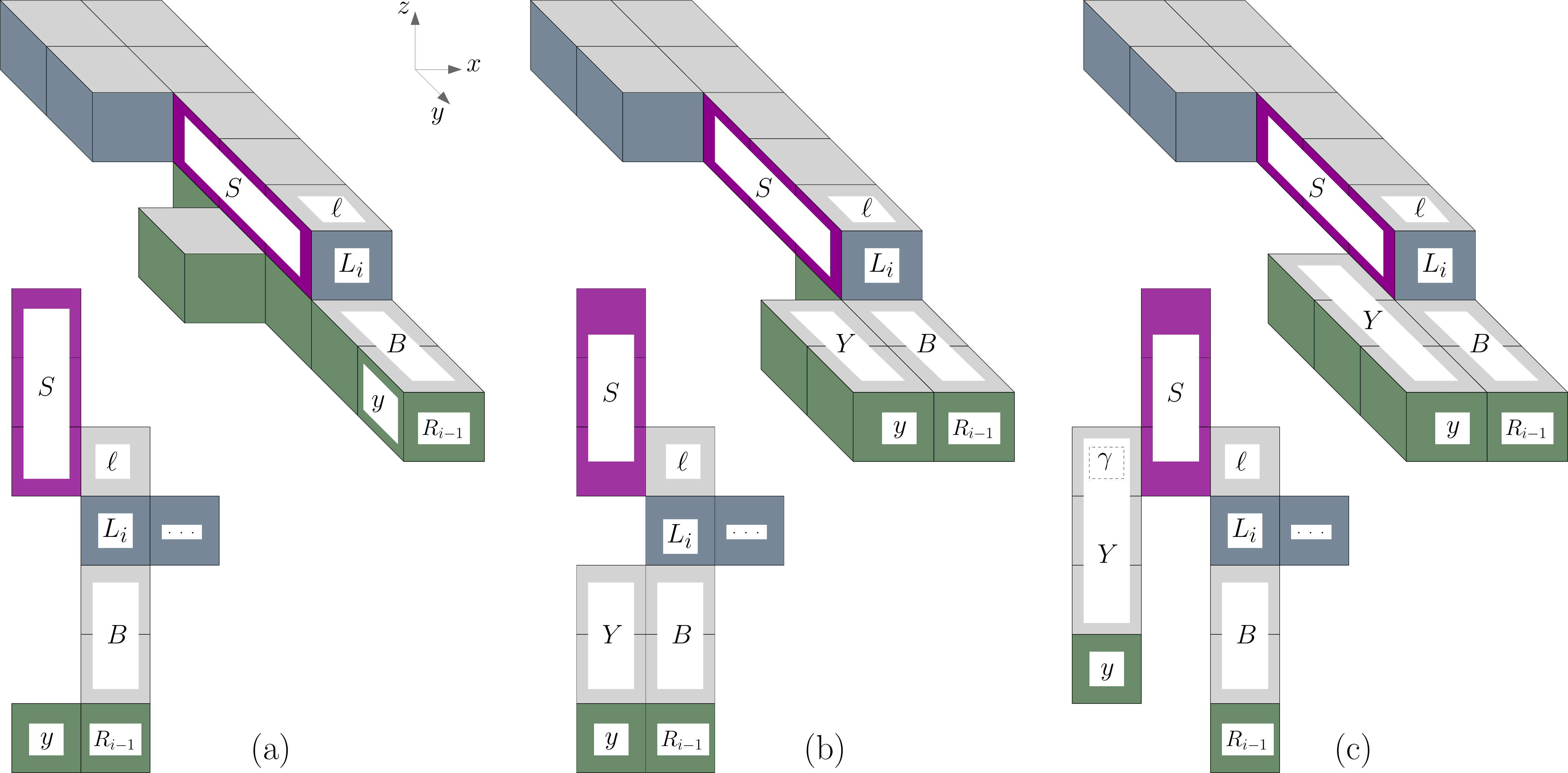}
    \caption{Common scenario 1: $\ell$ is attached to $L_i$ in the unfolding net; $B = (i-1)$-\beam($L_i$) is empty or on top of $\O_{i-1}$, and is identical to the $(i-1)$-bridge; $y$ is $cw$-adjacent to $R_{i-1}$; and $Y = (i-1)$-\beam($y$). (a) $Y$ is not in the partition of $\O_{i-1}$ (b) $Y$ is in the partition of $O_{i-1}$ but not adjacent to $S$ (c) $Y$ is in the partition of $O_{i-1}$,  adjacent to $S$.}
    \label{fig:unfops-common-2}
\end{figure} 

\vspace{1em}
Since $S$ is a left segment and $\O_i$ is orthogonally convex, $Y$ is empty or located on top of $\O_{i-1}$ (see~\cref{fig:unfops-common-2}). 
Recall that the top of $\O_{i-1}$ is partitioned into beams parallel to $(i-1)$-\beam($R_{i-1})$, which are parallel to $B$ (since $R_{i-1}$ is parallel to $L_i$, by definition). 
We attach $S$ to the unfolding net as follows:
\begin{itemize}
\item If $Y$ is not a beam in this partition---i.e., if $y$ is orthogonal to $R_{i-1}$, as shown in~\cref{fig:unfops-common-2}(a)---we attach $S$ to the left side of $\ell$. 
\item If $Y$ is a beam in this partition (meaning that $y$ is parallel to $R_{i-1}$), but it is not adjacent to $S$, we 
again attach $S$ to the left side of $\ell$. See~\cref{fig:unfops-common-2}(b). Note that this scenario is only possible if $B$ is non-empty. 
\item Otherwise, we cut the unfolding net along the vertical line separating $y$ and $R_{i-1}$, and then reconnect the two net components by attaching $S$ to both $Y$ and $\ell$. 
See~\cref{fig:unfops-common-2}(c). 
\end{itemize}

\begin{figure}[hpt]
    \centering
    \includegraphics[page=7,width=0.95\textwidth]{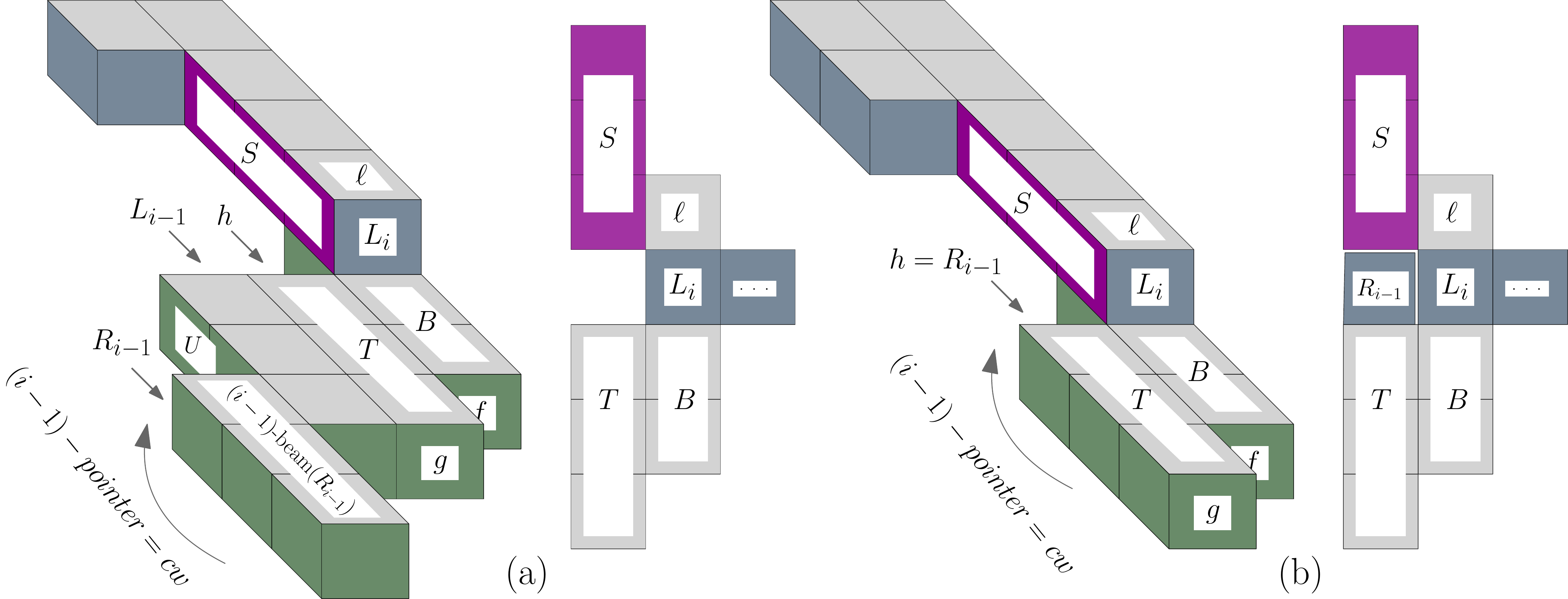}    
       \caption{Common scenario 2: $\ell$ is connected to $L_i$ in the unfolding net, $B$ is non-empty and on top of $\O_{i-1}$, the $(i-1)$-bridge consists of multiple beams, $T$ and $S$ are not adjacent, and (a) $h \neq R_{i-1}$ or (b) $h = R_{i-1}$.}
    \label{fig:unfops-common-4}
\end{figure}

\paragraph{Common scenario 2} In this scenario $\ell$ sits on top of $L_i$ in the unfolding net (after relocation, if necessary); 
$B$ is on the top of $\O_{i-1}$; and the $(i-1)$-bridge consists of multiple beams (one of which is $B$, which must be 
non-empty).
Let $T$ be the $(i-1)$-bridge beam adjacent to $B$. 
One side observation is that, if $y$ is parallel to $R_{i-1}$, then $T$ is identical to $Y$ (as shown in~\cref{fig:unfops-common-42}). 

The unfolding procedure for this scenario mirrors that of the unfolding scenario 1: 
if $T$ is not adjacent to $S$, attach $S$ to the left side of $\ell$, as shown in~\cref{fig:unfops-common-4}. Otherwise, 
split the $(i-1)$-bridge along the vertical line separating $T$ and $B$, then reconnect the two net components by attaching $S$ to both $T$ and $\ell$, as illustrated in~\cref{fig:unfops-common-42}(a).

\vspace{1em}
\begin{figure}[hpt]
    \centering
    \includegraphics[page=3,width=0.95\textwidth]{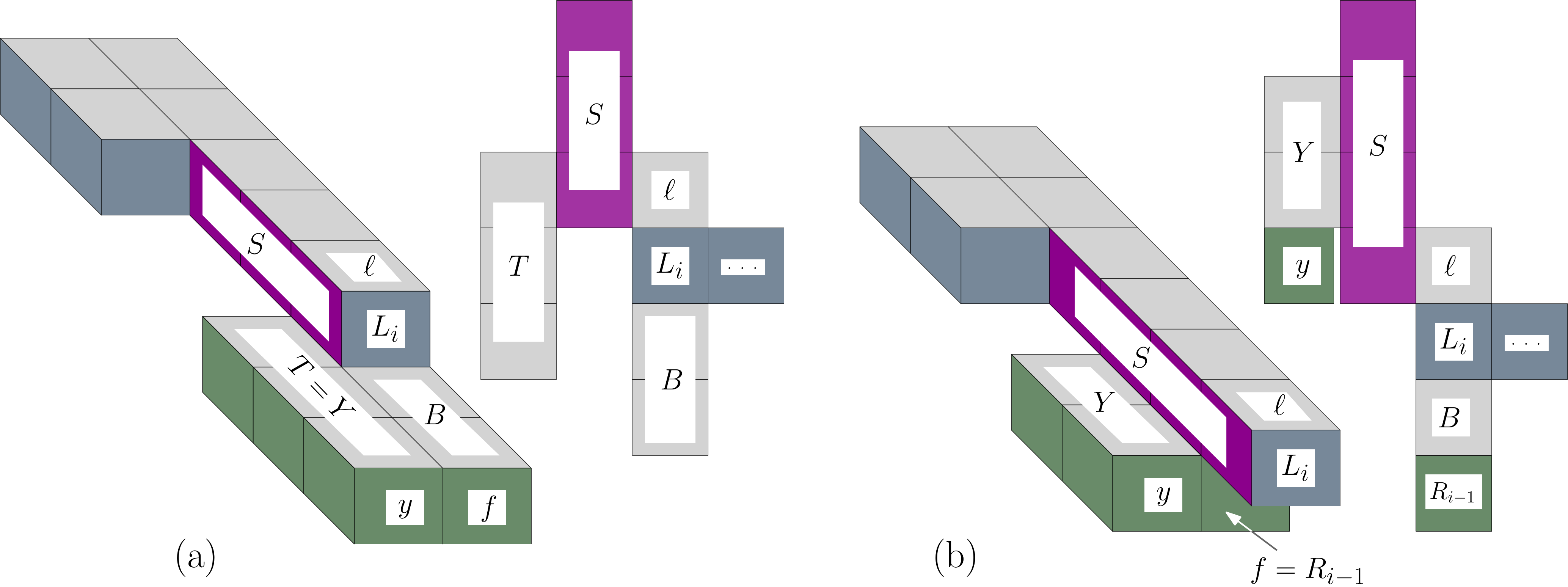}    
       \caption{(a) Common scenario 2: $\ell$ is connected to $L_i$ in the unfolding net, $B$ is non-empty and on top of $\O_{i-1}$, the $(i-1)$-bridge consists of multiple beams, and $T$ and $S$ are adjacent. (b) Common unfolding scenario 3: $\ell$ is connected to $L_i$ in the unfolding net, $B$ is non-empty and on the bottom of $\O_i$, and $Y$ is on top of $\O_{i-1}$.}
    \label{fig:unfops-common-42}
\end{figure}

\vspace{1em}
\paragraph{Common scenario 3} In this scenario $\ell$ sits on top of $L_i$ in the unfolding net (after relocation, if necessary), and 
$B$ is non-empty and on the bottom of $\O_i$. 
By~\cref{lem:bridge-bot}, $B$ is identical to the $(i-1)$-bridge. 
Consequently, $f = R_{i-1}$, $(i-1)$-pointer = $ccw$, and $y$ is attached to the left of $R_{i-1}$ in the unfolding net. 

We attach $S$ to the net as follows. If $Y$ is on top of $\O_i$, it must be adjacent to $S$, so the unfolding can proceed as in common scenario 1: cut the unfolding net along the vertical line separating $y$ and $R_{i-1}$, then reconnect the two resulting net 
components by attaching $S$ to $Y$ and $\ell$. See~\cref{fig:unfops-common-42}(b). 
If $Y$ is empty or on the bottom of $\O_i$, attach $S$ to the left side of $B$, as shown in~\cref{fig:unfops-common-1}. 

\vspace{1.5em}
\begin{figure}[htp]
    \centering
    \includegraphics[page=1,width=0.9\textwidth]{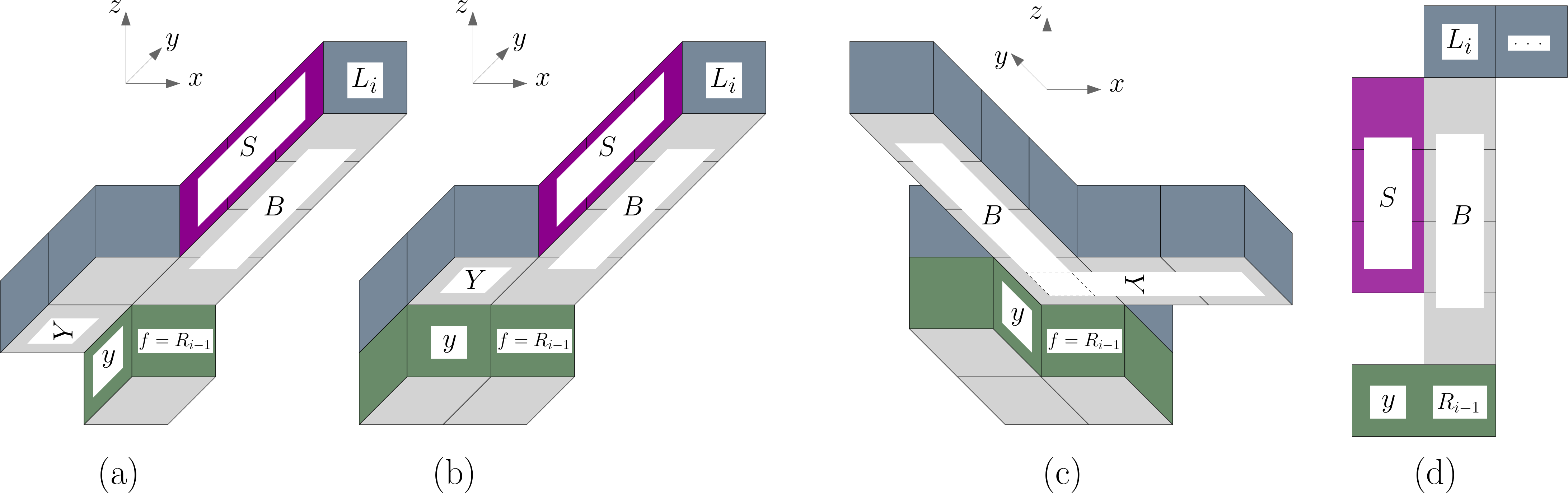}
    \caption{Common scenario 3: $B$ is non-empty and on the bottom of $\O_{i}$, $Y$ is on the bottom of $\O_i$, and $y$ faces (a) left, (b) front, or (c) right; (d) $S$ attaches to $B$ in the net.}
    \label{fig:unfops-common-1}
\end{figure}


\vspace{1em}
\paragraph{Common scenario 4} In this scenario $\ell$ cannot be relocated above $L_i$ in the unfolding net, without disconnecting the net. 
This can only happen if $i$-\beam($L_i$) is part of the $i$-bridge. By definition, $i$-\beam($R_i$) is also part of the $i$-bridge. 
By~\cref{lem:front-allback}, $L_i$ may not lie opposite to $R_i$, therefore $i$-\beam($L_i) \neq i$-\beam($R_i$). Consequently, 
the $i$-bridge consists of multiple beams. By~\cref{lem:bridge-visited}, all cells parallel to $R_i$ are visited. 
Since $R_i$ and $L_i$ are parallel (by the common scenario assumption), they must be connected by $S$. Refer to~\cref{fig:unfops-common-5}. 

\vspace{1em}
\begin{figure}[htp]
    \centering
    \includegraphics[page=5,width=0.95\textwidth]{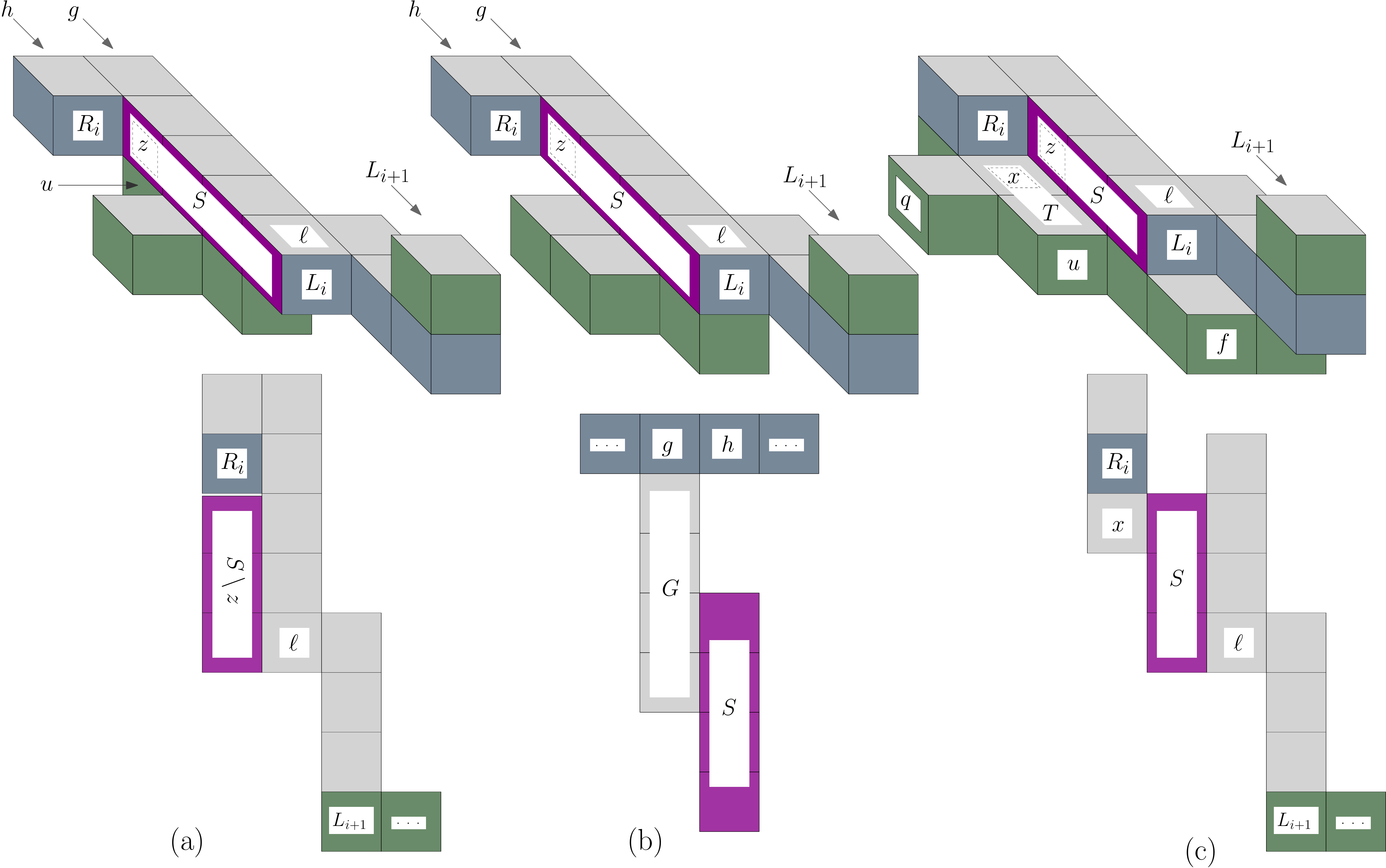}
    \caption{Common scenario 4: $\ell$ is disconnected from $L_i$ in the net, $z$ is $ccw$-adjacent to $R_i$, and $i$-\beam($z$) is (a) empty, (b) on the bottom of $\O_i$, 
    (c) on top of $\O_{i-1}$.}
    \label{fig:unfops-common-5}
\end{figure}

\vspace{1em}
Let $g$ and $h$ be the back $i$-band cells opposite $L_i$ and $R_i$, respectively. Note that both are visited, since $L_i$ and $R_i$ are front cells. 
Let $z \in S$ be the $i$-band cell adjacent to $R_i$. We unfold $S$ as follows:
\begin{itemize}
\squeezelist
\item If $(i-1)$-beam($z$) is empty, then $z$ is adjacent to an $(i-1)$-band cell $u$ with $(i-1)$-beam($u$) also empty. In this case, attach $z$ to $u$. If $S = z$, the unfolding is complete. Otherwise, relocate $R = (i-1)$-\beam($R_i$) by attaching it to $h$ in the unfolding net, 
and then attach $S \setminus z$ to the $i$-bridge. See~\cref{fig:unfops-common-5}(a). 
Since $S$ is a left segment connecting $R_i$ and $L_i$,  the piece $S \setminus z$ ends up directly below $R_i$, in the space previously held by $R$. 

\item If $(i-1)$-beam($z$) is on the bottom of $\O_i$, then $G = (i-1)$-beam($g$) is adjacent to $S$ (by the orthogonal convexity of $\O_{i-1}$). In this case, attach $S$ to $G$ in the unfolding net, as shown in~\cref{fig:unfops-common-5}(b)

\item If $(i-1)$-beam($z$) is on the top of $\O_{i-1}$, there is a top $(i-1)$-face cell $x$ adjacent to both $z$ and $R_i$. See~\cref{fig:unfops-common-5}(c). First, relocate $x$ to sit below $R_i$ in the unfolding net. Next, cut the $i$-bridge along the line separating $i$-\beam($R_i$) and $i$-\beam($L_i$). Finally, reconnect the two resulting net components by attaching $S$ to both $x$ and $\ell$.
\end{itemize}
\medskip
\noindent
This wraps up the common unfolding scenario. 

%% file: proof.tex
\section{Proof of Correctness}
\label{sec:proof}
This section is dedicated to proving the correctness of the unfolding algorithm described in~\cref{sec:unfolding}. Each proof section references the notation used in the corresponding algorithm section. 

Before delving into the proof details, we introduce two definitions. 
A cell $r$ or a beam $B$ in the unfolding net is \emph{up-vacant} (\emph{down-vacant}) if the vertical space above (below) it in the unfolding net, 
bounded by the lines passing through its vertical sides, is vacant. 
A cell $r$ in the unfolding net is termed \emph{relocatable} if its removal from the net does not disconnect it. 

The correctness of Stages 1 and 2 of the unfolding algorithm is straightforward. Stage 1 unfolds visited band segments horizontally and the connecting bridges vertically. By~\cref{lem:bridge-convex}, the result is a single connected piece that extends monotonically in  the horizontal direction without overlap. At the beginning of Stage 2, each visited $i$-band cell other than $R_i$ is up-vacant, so attaching a beam $\beta$ above a visited anchor $r \neq R_i$ does not result in overlap. If $\beta$'s other anchor $q$ is an unvisited $i$-band cell, attaching $q$ above $\beta$ simply extends $\beta$ vertically without creating overlap. Next we establish the correctness of Stages 3 and 4, which require more complex reasoning. 

\subsection{Proof of Correctness --- Stage 3} 
\label{sec:proof-stage3}
In this section we prove the correctness of the unfolding process for the bottom surface of $\O_i$, for a fixed $i \ge 1$, described in~\cref{sec:bottom}. 
At the beginning of this stage, all visited $i$-band cells in the unfolding net, except for $L_i$, are down-vacant. 
Also, any unvisited i-band cell $r$ parallel to $R_i$ sits on top of $i$-\beam($r$) in the unfolding net (as described in~\cref{sec:top}). 
Our unfolding procedure attaches a bottom beam $\beta$ either:
\begin{enumerate}
\squeezelist
\item [(a)] below a visited $i$-band anchor $r$, or 
\item [(b)] above an unvisited $i$-band anchor $r$ parallel to $R_i$, or 
\item [(c)] above or below a visited $(i-1)$-band anchor $q$, depending on the $ccw$ or $cw$ direction of the  $(i-1)$-pointer. 
\end{enumerate}

\noindent
Attaching $\beta$ below a visited $i$-band anchor $r$ does not create an overlap, since $r$ is down-vacant. 
Attaching $\beta$ above an unvisited $i$-band anchor $r$ simply extends an existing vertical strip vertically, so no overlap occur. 
Now consider the case where $\beta$ gets attached to a visited $(i-1)$-band anchor $q$. 
Since $\beta = (i-1)$-\beam($q$) is not on the surface of $\O_{i-1}$, $q$ is not attached to any strip in the top partition 
of $\O_{i-1}$ (see discussion in~\cref{sec:top}). 
This implies that the place where $\beta$ would attach to $q$ in the unfolding net (above or below $q$, depending on the $ccw$ or $cw$ direction of the  $(i-1)$-pointer) is vacant. This allows us to attach $\beta$ to $q$ without causing overlap. 

In the case where the $i$-band is not nearly-visited, subsequently attaching to $\beta$ its unvisited $i$-band anchor $r$ simply extends $\beta$ vertically without creating overlap. Similar arguments apply for the case where an unvisited $i$-band cell $u$ attaches to an adjacent $(i-1)$-band cell $v$. 

\medskip
\noindent
The following definition is used in~\cref{lem:net-properties} to establish properties of the unfolding net at the end of Stage 3.

\begin{definition}
A beam $\beta$ on the bottom of $\O_i$ anchored on the $(i-1)$-band is called \emph{stable} if it either does not intersect the $(i-1)$-bridge or, if it does, the remaining portion---after removing the shared section---is empty or consists of a single sub-beam anchored on the $i$-band. 
Otherwise, $\beta$ is called \emph{unstable}.
\label{def:stable}
\end{definition}

For example, in the scenario depicted in~\cref{fig:ortho-bottom}, $(i-1)$-\beam($R_i$) is unstable, because the portion not shared with the $(i-1)$-bridge consists of two separate components, $\beta_1$ and $\beta_2$. In contrast, $(i-1)$-\beam($u$) is stable, because the section not shared with the $(i-1)$-bridge forms a single component (labeled $\delta_1$) anchored on the $i$-band. 
Our unfolding algorithms maintains stable $(i-1)$-beams attached to the $i$-band, whereas unstable $(i-1)$-beams get split between the $i$-band and the $(i-1)$-band.

\medskip
\noindent
Next, we prove several key properties of the current unfolding net, which are essential to proving the correctness of the subsequent unfolding steps.

\vspace{1em}
\begin{lemma} [{\bf Net Properties}] At the end of Stage 3, the unfolding net satisfies the following properties, for each $i$:
\begin{enumerate}[leftmargin=3.5em]
\item [{\boldmath $\P1(i)$}:] For any visited $i$-band cell $r \neq L_i$ anchored on $B = (i-1)$-\beam($r$), if one of the following holds:
\begin{itemize}
\squeezelist
\item [(a)] $B$ is empty, or non-empty but not in the bottom partition of $\O_i$, or 
\item [(b)] $B$ is in the bottom partition of $\O_i$, does not cross the $(i-1)$-bridge, and both of its anchors lie on the $i$-band. If both anchors are visited, then $r$ is first visited. Additionally, if the $i$-band is nearly-visited, both anchors are visited.
\end{itemize}
then $r$ is down-vacant (up-vacant) based on the $ccw$ ($cw$) direction of the $i$-pointer. 
%
\item [{\boldmath $\P2(i)$}:] For any visited $i$-band cell $r \neq R_i$, if $B = i$-\beam($r$) is empty, 
then $r$ is up-vacant (down-vacant), unless attached to an unvisited $(i+1)$-band cell $u$ 
when the $(i+1)$-band is nearly-visited, in which case $u$ is up-vacant 
(down-vacant), based on the $ccw$ ($cw$) direction of the $i$-pointer. 
If $B$ is not empty and if either:
\begin{itemize}
\squeezelist
\item [(a)] $B$ is on top of $\O_i$ but not in the partition of $\O_i$, or 
\item [(b)] $B$ is a stable beam on the bottom of $\O_{i+1}$ and, if the $(i+1)$-band is nearly-visited, 
$B$'s anchor on the $(i+1)$-band is visited, 
\end{itemize}
then $r$ is up-vacant (down-vacant) based on the $ccw$ ($cw$) direction of the $i$-pointer. Moreover, if:
\begin{itemize}
\squeezelist
\item [(c)] $B$ is in the top partition of $\O_i$, not anchored on $L_{i+1}$ or an unvisited $i$-band cell, 
\end{itemize}
then $B$ is up-vacant (down-vacant) based on the $ccw$ ($cw$) direction of the $i$-pointer. 
%
%
\item [{\boldmath $\P3(i)$}:] For any visited $i$-band cell $r$ adjacent to an $(i+1)$-band cell $u$, the cell space immediately above (below) $r$, 
based on the $ccw$ ($cw$) direction of the $i$-pointer, is either vacant or occupied by $u$. 

\item [{\boldmath $\P4(i)$}:]  If an $i$-bridge beam $B$ is not anchored on $L_{i+1}$, then $B$ is up-vacant (down-vacant) 
based on the $ccw$ ($cw$) direction of the $i$-pointer. 
\item [{\boldmath $\P5(i)$}:]  If an $i$-bridge beam $B$ is not anchored on $R_i$, then $B$ is down-vacant (up-vacant) 
based on the $ccw$ ($cw$) direction of the $i$-pointer. 
\item [{\boldmath $\P6(i)$}:] {\bf Relocation.} If $B = (i-1)$-\beam($R_i$) is not part of the $(i-1)$-bridge:
\begin{itemize}
\squeezelist
\item [(a)] If $B$ is in the top partition of $\O_{i-1}$ and $x \in B$ adjacent to $R_i$, and if either $B \setminus x$ is non-empty 
or $B$ is not anchored on $L_{i-1}$, then $x$ is relocatable.
\item [(b)] If $B$ is in the bottom partition of $\O_i$ and both of its anchors are on the $i$-band, then $B$ is relocatable. 
\item [(c)] Moreover, if $Y$ is a beam in the top partition of $\O_i$, not part of the $i$-bridge and not anchored on $L_i$, and if both of its anchors are visited $i$-band cells, then $Y$ is relocatable. 
\end{itemize}

\end{enumerate}
\label{lem:net-properties}
\end{lemma}
\begin{proof}
We prove these properties for arbitrary $i \ge 0$. 
When $i$-pointer = $ccw$, the top and bottom edges of the $i$-band correspond to top and bottom edges in the unfolding net, respectively. As a result, the top (bottom) surface of $\O_i$ attaches above (below) the $i$-band in the unfolding net.
%
When $i$-pointer = $cw$, the bottom edges of the $i$-band correspond to the top edges in the unfolding net, so the top (bottom) surface of $\O_i$ attaches below (above) the $i$-band in the unfolding net. We prove this lemma for the $ccw$ direction of the $i$-pointer; the $cw$ direction is symmetric.

We first prove the net property $\P1$. Let $r \neq L_i$ be a visited $i$-band cell, and let $B = (i-1)$-\beam($r$). 
Note that, if non-empty, $B$ can be either on top of $\O_{i-1}$ or on the bottom of $\O_i$. 
\begin{itemize}
\squeezelist
\item If on top of $\O_{i-1}$, then $B$ is not a beam in the partition of $\O_i$. Furthermore, $B$ does not get attached to $r$ in the unfolding net, because the top of $\O_{i-1}$ exclusively attaches to the $(i-1)$-band, as discussed in Stage 2 (\cref{sec:top}).
\item Otherwise, if on the bottom of $O_i$ but not in the partition of $\O_i$, then $B$ is not one of the beams processed in Stage 3 (\cref{sec:bottom}) and therefore it does not get attached to $r$ in the unfolding net. 
\item Otherwise, $B$ is a beam in the bottom partition of $\O_i$. Assume that $B$ satisfies the conditions of $\P1b(i)$: it does not cross the $(i-1)$-bridge,  its second anchor $q \neq r$ is also on the $i$-band, and if $q$ is visited, it is visited after $r$. Furthermore, if the $i$-band is nearly-visited, $q$ must be visited.

If $q$ is unvisited, then the $i$-band cannot be nearly-visited. In this case, $r$, $i$-\beam($r$), $q$ and $B$ extend vertically above $r$ in that order, ensuring that $r$ remains down-vacant.  

If $q$ is visited, then $B$ attaches to the \emph{second} visited anchor $q$ (not $r$) in Stage 3 (\cref{sec:bottom}). Hence, $r$ again remains down-vacant. 
\end{itemize}
In all  scenarios, $r$ is down-vacant. Therefore, property {\boldmath $\P1(i)$} holds. 

Next, we prove property $\P2$. Let $r \neq R_i$ be a visited $i$-band cell, and let $B = i$-\beam($r$). 
If $B$ is empty, then $r$ is adjacent to an $(i+1)$-band rectangle $u$. 
In this case, $r$ is up-vacant unless the $(i+1)$-band is nearly-visited and $u$ is unvisited, causing $u$ to attach directly above $r$ 
(see cell labeled $u$ in~\cref{fig:ortho-ex-complete}). 
However, nothing else attaches above $u$, so property {\boldmath $\P2(i)$} holds in this case. 

Assume now that $B$ is non-empty. Then $B$ is located either on top of $\O_{i}$ or on the bottom of $\O_{i+1}$: 
\begin{itemize}
\squeezelist
\item If $B$ is on top of $\O_i$ but not in the partition of $\O_i$, then it was not processed in Stage 2 (\cref{sec:top}) and does not attach to $r$ in the unfolding net. Therefore, $r$ remains up-vacant, and property {\boldmath $\P2a(i)$} holds. 
\item Consider now the case where $B$ is a \emph{stable} beam located on the bottom of $\O_{i+1}$. 
If $B$ is not in the partition of $\O_{i+1}$, then it does not attach to $r$ in the unfolding net. 
Otherwise, by~\cref{def:stable}, the portion $B' \subseteq B$ not shared with the $i$-bridge is a single component anchored on the $(i+1)$-band, and  attaches to that band in Stage 3 (\cref{sec:bottom}). Thus, $B$ does not attach to $r$ in the unfolding net, $r$ remains up-vacant, 
and property {\boldmath $\P2b(i)$} holds. 
\item Finally, assume that $B$ is a top beam in the partition of $\O_i$.  
In Stage 2 (\cref{sec:top}), the second anchor $r' \neq r$ of $B$ attaches above $B$ only if either $r'$ is an unvisited $i$-band cell or $r' = L_{i+1}$. Since  these scenarios are excluded by {\boldmath $\P2c(i)$}, $B$ remains up-vacant and property {\boldmath $\P2c(i)$} holds. 
\end{itemize}

\noindent
Regarding property $\P3$, if $r$ is adjacent $(i+1)$-band cell $u$, then $u$ attaches above $r$ 
only if it is unvisited and the $(i+1)$-band is nearly-visited.  
In all other cases, nothing attaches above $r$, therefore property {\boldmath $\P3(i)$} holds. 

%
Properties {\boldmath $\P4(i)$} and {\boldmath $\P5(i)$} follow immediately from the fact that the only band cells attached to the $i$-bridge in the unfolding net are $R_i$ and $L_{i+1}$.

Finally, we turn to property  {\boldmath $\P6(i)$}. Consider $B = (i-1)$-\beam($R_i$) and assume that $B$ is not part of the $(i-1)$-bridge. If $B$ is in the top partition of $\O_{i-1}$, then $B$ is anchored on an $(i-1)$-band cell $r$ (cf.~\cref{lem:beam-ends}). By~\cref{lem:liri-out}, $r$ is visited. Our unfolding procedure attaches $B$ to $r$ in the unfolding net in Stage 2 (\cref{sec:top}). Since $B$ is not anchored on an unvisited $(i-1)$-band cell, no other surface pieces get attached to $B$, so removing $B$ from the net does not disconnect it. This implies that the entire beam $B$ is relocatable, and therefore $x \in B$ adjacent to $R_i$ is also relocatable. So property {\boldmath $\P6a(i)$} holds. The additional constraints on $B$ are required in the inductive proof for Stage 4. 

If $B$ is in the bottom partition of $\O_i$ and both of its anchors are on the $i$-band, then both anchors are visited, as one of them is $R_i$ and the $i$-band cell opposite $R_i$ is visited (cf.~\cref{lem:front-allback}). In this situation, $B$ gets attached to $R_i$ in Stage 3 (\cref{sec:bottom}), and nothing else gets attached to $B$ thereafter. Therefore, $B$ is relocatable and property {\boldmath $\P6b(i)$} holds. Similarly, if $Y$ is a top beam in the partition of $\O_i$ but not part of the 
$i$-bridge, and both of its anchors are visited $i$-band cells other than $L_i$, then $Y$ gets attached to the second visited anchor in Stage 2 (\cref{sec:top}). 
Furthermore, no other surface pieces get attached to $Y$. It follows that $Y$ is relocatable, thus establishing property {\boldmath $\P6c(i)$}. 
\end{proof}

\subsection{Proof of Correctness --- Stage 4}
\label{sec:proof-stage4}
In this section we prove the correctness of Stage 4 of the unfolding process described in~\cref{sec:allbands}. 
Two non-overlapping surface pieces of $\O$ are considered \emph{incident} if their boundaries share at least one cell vertex.
Note that while adjacent surface pieces (which share a cell edge) are necessarily incident, the converse does not always hold: 
incident pieces may share only a vertex and not an edge. %
In Stage 4, the unfolding algorithm completes the unfolding of $\O_i$ incrementally, for increasing $i$. Recall the assumption that $R_i$ is a front face and the $i$-pointer is $ccw$. 
We will show by induction on $i$ that our unfolding process maintains the following local invariant. 

\newpage
\noindent
{\bf Local Invariant $\I(i)$:} Prior to processing $\O_{i}$, the following net properties hold: 

\medskip
\noindent
\begin{tabular}{ll}
$\I1(i): ~\P1(j)$ for &  ($j = i-1$, $r = R_{i-1}$, and the $(i-1)$-band is nearly-visited),  \\
 & ($j = i$ and $r$ not adjacent to $L_i$), and ($j > i$). \\
$\I2(i):~\P2(j)$ for &  ($j = i-1$ and $r$ visited just before $R_{i-1}$), and ($j \ge i$). \\
$\I3(i):~\P3(j)$ for &  ($j \ge i-1$). \\
$\I4(i):~\P4(j)$ for &  ($j \ge i-1$). \\
$\I5(i):~\P5(j)$ for &  ($j = i-1$ and $R_{i-1}$ adjacent to $L_{i-1}$),  \\
				& ($j = i-1$ and $(i-1)$-\beam($R_{i-1}$) not incident to $L_{i-1}$), \\
				& and ($j \ge i$). \\
$\I6(i):~\P6(j)$ for &  ($j \ge i$). \\
\end{tabular}

\medskip
\medskip
\noindent
Note that, as we advance to higher layers, the invariant imposes stricter restrictions on lower levels. 
The second applicability constraint for $\I5$ may seem subtle, but it is straightforward: 
it simply excludes the case where $(i-1)$-\beam($R_{i-1}$) is incident to $L_{i-1}$. This 
occurs when $L_{i-1}$ and $R_{i-1}$ share the same normal, 
are connected by a straight $i$-band segment, and 
$R_{i-1}$ is in front of $L_{i-1}$ when looking into the normal for $R_{i-1}$. 
In particular, $\I5(i+1)$ excludes the case depicted in~\cref{fig:unfops-4}. 

By~\cref{lem:net-properties}, at this point the invariant $\I(i)$ holds, for each $i$. 
Since $R_1$ and $L_1$ are adjacent, the entire surface of $\O_1$ is part of the unfolding net (as discussed in~\cref{sec:bands-bridges}).  
The unfolding algorithm incrementally completes the unfolding of $\O_i$, for $i = 2, 3, \ldots$. For each $i$, we use the Invariant $\I(i)$ and 
prove that, upon completing the unfolding of $\O_i$, the result is a single, non-overlapping, connected net, and the Invariant $\I(i+1)$ continues to hold. 

\subsubsection{Proof of Correctness --- Stage 4, Case 1. }  
This section references the notation used in~\cref{sec:case1} and~\cref{fig:unfops-1}. 
The context for this case is that both $a$ and $b$ are unvisited. 
Since $S'$ is orthogonal to $R_i$ and the top of $\O_i$ is partitioned into beams parallel to $i$-\beam($R_i)$, each cell in $S'$ satisfies the conditions of the net property $\P2a(i)$. 
By Invariant $\I2(i)$ this property holds, therefore each cell in $S'$ is up-vacant. 
If $S$ is no longer than $S'$, then $S$ fits within the vacant space above $S'$ without overlap. 
If $S$ is longer than $S'$, the unfolding process merely shifts $i$-\beam($a'$) and everything to its right further rightward, 
thus avoiding any overlap. 

At this point, the Invariant $\I2(i)$ no longer holds for the $i$-band cells in $S'$ --- however, the subsequent unfolding process for $\O_i$ does not interact with this particular net piece. Furthermore, since $R_i$ is a front cell (by assumption), none of the band cells $r \in S'$ may be visited just before $R_i$, therefore invariant $\I2(i+1)$ is not applicable to $r$. Regarding relocation:
\begin{itemize}
\squeezelist
\item  Neither $i$-\beam($a$) nor $i$-\beam($b$) is anchored solely on visited $i$-band cells, 
so property $\P6c(i)$ from Invariant $\I6(i)$ does not apply to either beam. 
\item Neither beam is anchored on $R_{i+1}$, so properties $\P6a(i+1)$ and $\P6b(i+1)$ from Invariant $\I6(i+1)$ also do not apply.
\end{itemize}
All other net properties remain unaffected. Consequently, the invariant $\I(i+1)$ continues to hold. 

\subsubsection{Proof of Correctness --- Stage 4, Case 2. }  
This section references the notation used in~\cref{sec:case2}. 
The context for this case is that $a$ is visited and $b$ is unvisited. 
Consider first the case where $S$ is a right segment. Refer to~\cref{fig:unfops-2}. 
Since $i$-\beam($S$) is non-empty and on top of $\O_i$, $S$ is adjacent to, and no longer than, $i$-\beam($a$). 
See~\cref{fig:unfops-2}. 
Let $x$ be the next visited $(i+1)$-band cell in $ccw$ order from $L_{i+1}$, adjacent to it.
We show that $x$ is down-vacant. Consequently, $S$ fits below $x$ in this vacant space, without overlap. 
Let $X = i$-\beam($x$).

\vspace{0.5em}
\begin{figure}[htp]
    \centering
    \includegraphics[page=2,width=0.95\textwidth]{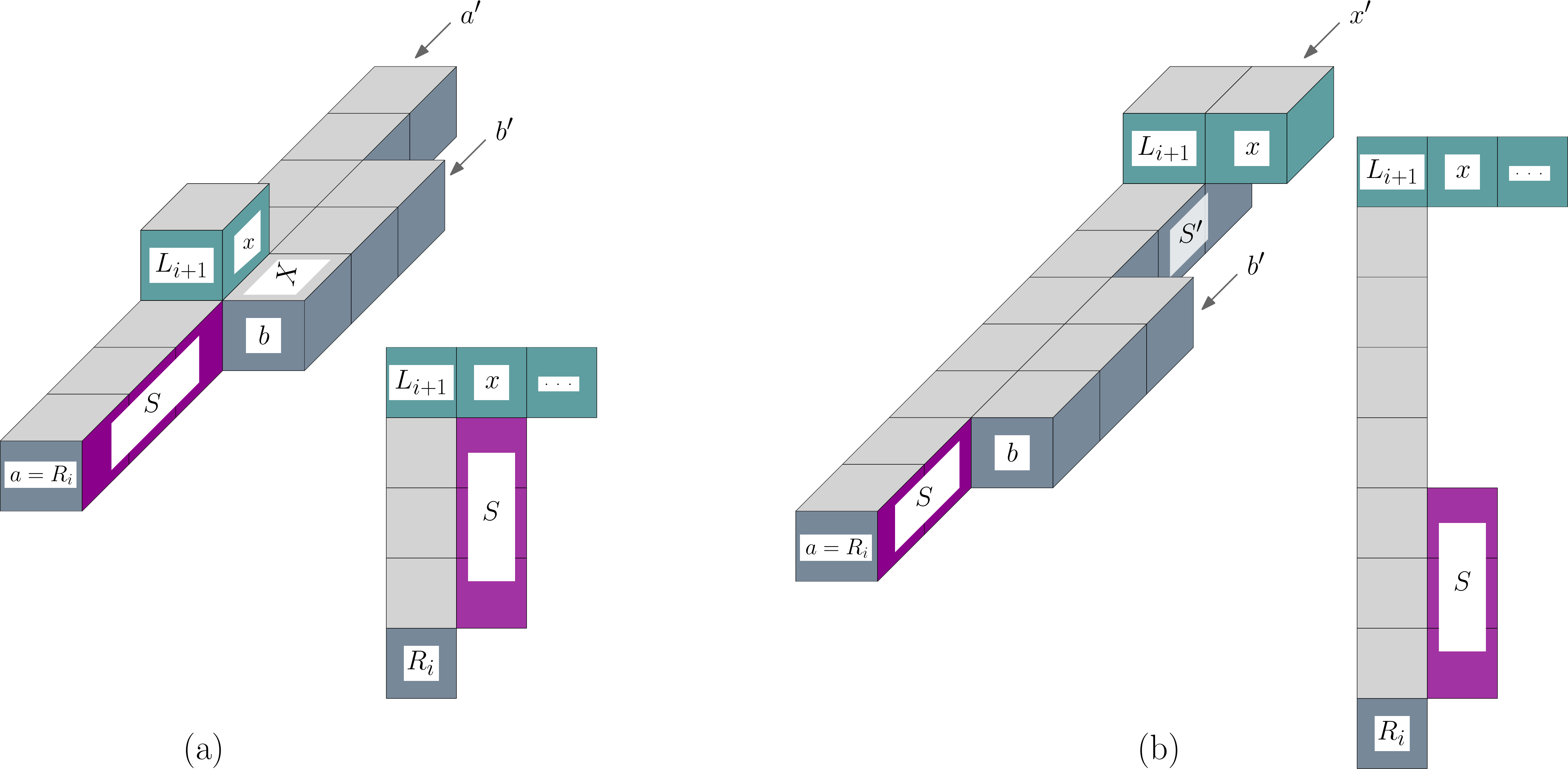}
    \caption{Proof of correctness for stage 4, case 2: $S$ is a \emph{right} segment and $i$-\beam($x$) is (a) on top of $\O_i$ (b) on the bottom of $\O_{i+1}$.}
    \label{fig:unfops-2}
\end{figure}

\vspace{0.5em}
If $X$ is empty, or non-empty but not in the bottom partition of $\O_{i+1}$ (as depicted in~\cref{fig:unfops-2}(a)), the net property $\P1a(i+1)$ guaranteed by the Invariant $\I1(i)$ tells us that $x$ is down-vacant. 
Otherwise, $X$ is non-empty and in the bottom partition of $\O_{i+1}$. By the $2$-manifold property of $\O$'s surface, $x$ must be coplanar with $L_{i+1}$. 
This, along with the orthogonally convex property of $\O_i$, implies that $X$ is anchored on the $(i+1)$-band cell $x'$ opposite $x$. Refer to~\cref{fig:unfops-2}(b). 
Since the $i$-bridge is on a top face, it cannot intersect $X$. 
Also note that, if both $x$ and $x'$ are visited, then $x$ is first visited. Furthermore, if the $(i+1)$-band is nearly visited, then $R_{i+1}$ is separated by straight band segment from $L_{i+1}$ and therefore must lie $ccw$ from $x'$, which implies that $x'$ is visited. 
Thus, the net property $\P1b(i+1)$ guaranteed by the Invariant $\I1(i)$ applies to $x$, telling us that $x$ is down-vacant.

At this point, the net property $\P1(i+1)$ claimed by Invariant $\I1(i)$ no longer holds for $x$ --- however, the subsequent unfolding process for $\O_i$ does not interact with this particular net piece. Furthermore, since $x$ is adjacent to $L_{i+1}$, Invariant $\I1(i+1)$ is not applicable to $x$, and since 
$x \neq R_{i+1}$, Invariant $\I1(i+2)$ is not applicable to $x$. Regarding relocation:
\begin{itemize}
\squeezelist
\item The beam $i$-\beam($a$) used to attach $S$ is not anchored on $R_{i+1}$, 
so properties $\P6a(i+1)$ and $\P6b(i+1)$ from Invariant $\I6(i+1)$ do not apply to it.
\item This beam is not anchored solely on visited $i$-band cells (since it is anchored on $L_{i+1}$), 
so property $\P6c(i)$ from Invariant $\I6(i)$ does not apply to it. 
\end{itemize}
All other net properties remain unaffected. Consequently, the invariant $\I(i+1)$ continues to hold.

\medskip
Consider now that case where  $S$ is a \emph{left} segment. Refer to~\cref{fig:unfops-3}. Since $b$ is unvisited, $b'$ must be visited (cf.~\cref{lem:strips}). Furthermore, $i$-\beam$(b') \equiv i$-\beam($b$) sits on top of $b'$ in the unfolding net (as described in~\cref{sec:top}). Let $x$ be the $i$-band cell adjacent to $b'$ in the $ccw$ direction. 
Our unfolding procedure attaches $S$ to $i$-\beam($b'$), thereby positioning it in the vertical space above $x$. We must show that no overlap occurs. 

\vspace{1em}
\begin{figure}[htp]
    \centering
    \includegraphics[page=3,width=0.99\textwidth]{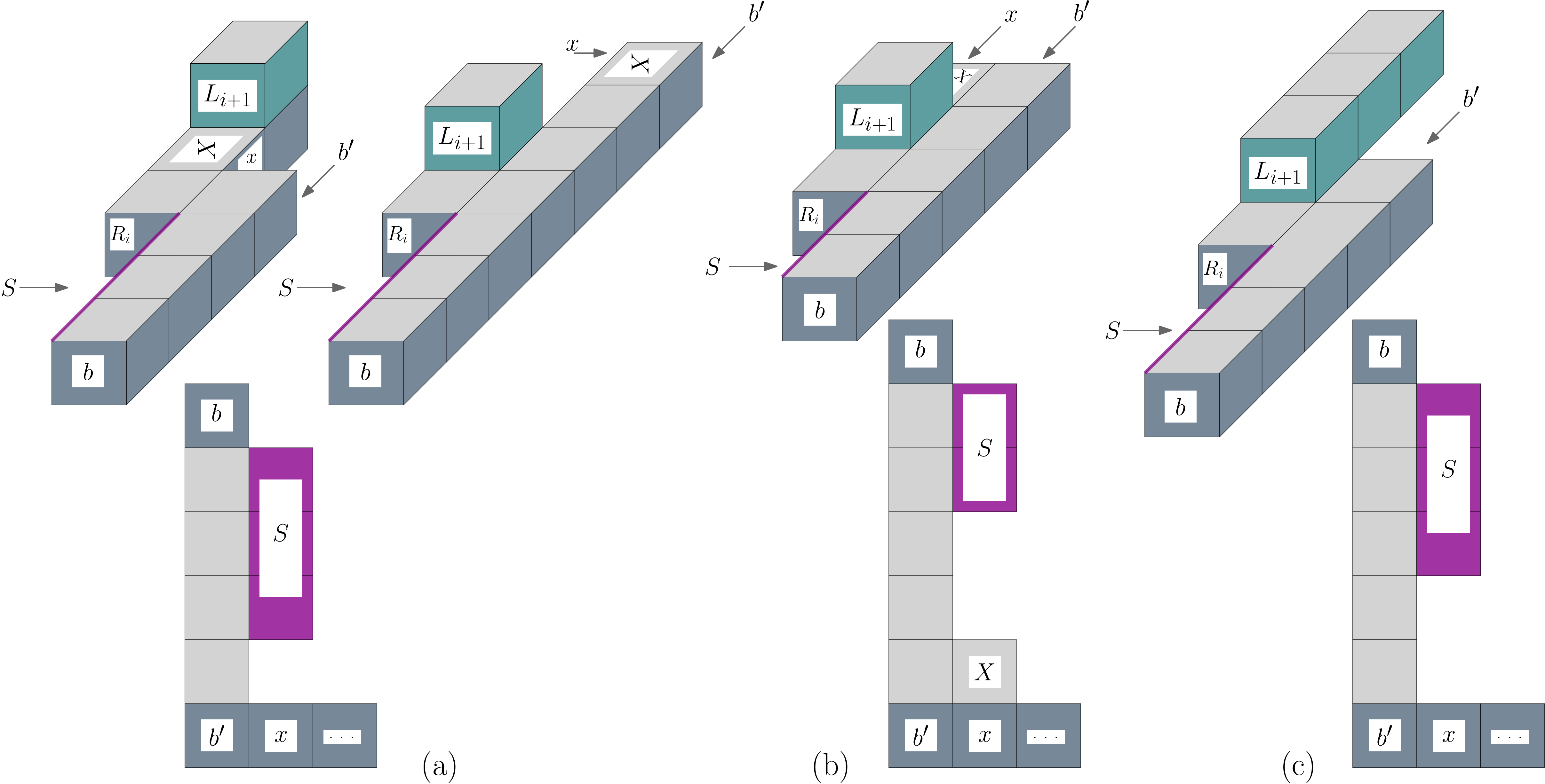}
    \caption{Proof of correctness for stage 4, case 2:: $S$ is a \emph{left} segment and $x$ is (a) orthogonal to $b'$, or (b) parallel to $b'$.}
    \label{fig:unfops-3}
\end{figure}

\vspace{1em}
 Let $X = i$-\beam($x$). If $X$ is empty, $x$ must be adjacent to an $(i+1)$-band rectangle $u$. 
By the net property $\P2(i)$ guaranteed by Invariant $\I2(i)$, either $x$ is up-vacant, or $u$ sits on top of $x$ in the unfolding net, and $u$ is up-vacant. 
Given that $\O$'surface is a 2-manifold, $x$ cannot share a vertex with $S$. Therefore, in the unfolding net,
there is at least one cell's worth of vertical space between $x$ and $S$, ensuring that no overlap occurs. 
Note that the property $\P3(i)$ claimed by Invariant $\I3(i)$ remains unaffected, since it involves only the cell immediately above $x$. 
 Assume now that $X$ is non-empty. We discuss two scenarios:
\begin{enumerate}
\item $x$ is orthogonal to $R_i$ (and so orthogonal to $b'$). Refer to~\cref{fig:unfops-3}(a). In this case $X$ must be on top of $\O_i$ (because the surface of $\O$ is a 2-manifold). Note that $X$ is not one of the beams in the partition of $\O_i$, because the the top of $\O_i$ is partitioned into beams parallel to $i$-\beam($R_i$), and $x$ is orthogonal to $R_i$. 
By the net property $\P2a(i)$ guaranteed by Invariant $\I2(i)$, $x$ is up-vacant, so $S$ fits above $x$ without overlap. 
\item $x$ is parallel to $R_i$. Assume first that $X$ is on top of $\O_i$, as depicted in~\cref{fig:unfops-3}(b). In this case, $X$ sits on top of $x$ in the unfolding net (as described in~\cref{sec:top}). However, $X$ is shorter than the vertical distance between $x$ and $S$ (since $S$ is a left segment), therefore it does not overlap with $S$.  Since $X$ is not anchored on $L_{i+1}$ or an unvisited $i$-band cell, the net property $\P2c(i)$ guaranteed by Invariant $\I2(i)$ applies to $X$. This tells us that $X$ is up-vacant, therefore no overlap occurs.

Assume now that $X$ is on the bottom of $\O_{i+1}$, as depicted in~\cref{fig:unfops-3}(c). 
In this case, $X$ may or may not be in the partition of $\O_{i+1}$. 
In either case, note that $X$ does not cross the $i$-bridge, since the $i$-bridge is on a top face. By~\cref{def:stable}, $X$ is stable.  
If the $(i+1)$-band is nearly-visited, all back $(i+1)$-band cells must be visited.  
Consequently, $X$'s anchor on the $(i+1)$-band is visited. 
In this case, the net property $\P2b(i)$ guaranteed by Invariant $\I2(i)$ applies to $x$, telling us that $x$ is up-vacant. 
Therefore, $S$ fits above $x$ without overlap. 
\end{enumerate} 
At this point, the Invariant $\I2(i)$ no longer holds for $x$ --- however, the subsequent unfolding process for $\O_i$ does not interact with this specific net piece. Furthermore, since $x$ is not visited just before $R_i$ (because $x$ is a back face), Invariant $\I2(i+1)$ is not applicable to $x$. 
Regarding relocation: 
\begin{itemize}
\squeezelist
\item The beam $i$-\beam($b'$) used to attach $S$ is not anchored on $R_{i+1}$, 
so properties $\P6a(i+1)$ and $\P6b(i+1)$ from Invariant $\I6(i+1)$ do not apply.
\item Since $b$ is unvisited, $i$-\beam($b'$) is not anchored solely on visited $i$-band cells, 
so property $\P6c(i)$ from Invariant $\I6(i)$ does not apply to this beam. 
\end{itemize}
All other net properties remain unaffected. Consequently, the invariant $\I(i+1)$ continues to hold. 

\subsubsection{Proof of Correctness --- Stage 4, Common scenario 1.}   
This and the next three sections reference the notation used in~\ref{sec:appendix-common}: 
$L = i$-\beam($L_i$), $B = (i-1)$-\beam($L_i$), $f$ is the front $(i-1)$-band cell anchored to $B$ (or to $L_i$ if $B$ is empty), 
$y$ is the $(i-1)$-band cell $cw$ from $f$, $Y = (i-1)$-\beam($y$), and $y' \neq y$ is the other anchor of $Y$. 
Also, the the $i$-pointer is $ccw$ and the $i$-band is not nearly-visited. 
Refer to~\cref{fig:unfops-common-2}. 
The unfolding process for this scenario begins, if necessary, by relocating the cell $\ell \in L$ adjacent to $L_i$. 
This relocation preserves the local invariant $\I(i)$, since $\ell$ is an endpoint of $L$ and nothing else is attached to it. 

The context for common scenario 1 is that $\ell$ sits on top of $L_i$ in the unfolding net, $B$ is empty or on top of $\O_{i-1}$, the $(i-1)$-bridge is $B$, and 
$f = R_{i-1}$ . 
Note that Invariant $\I2(i)$ applies to the $(i-1)$-band cell $y$, since $y$ is visited just before $R_{i-1}$. %
If $Y$ is empty, then $y$ meets the conditions of the net property $\P2(i-1)$ guaranteed Invariant $\I2(i)$. 
Similarly, if $Y$ is not in the partition of $\O_{i-1}$ (i.e., $y$ is orthogonal to $R_{i-1}$),
then $y$ meets the conditions of the net property $\P2a(i-1)$ guaranteed by Invariant $\I2(i)$.
In both cases, this property indicates that $y$ is up-vacant. Consequently, $S$ fits in the space above $y$ without overlap. 
At this point, the Invariant $\I2(i)$ no longer holds for $y$ --- however, the subsequent unfolding process for $\O_i$ 
does not interact with this specific net piece. 
Note that the cell space immediately above $y$ remains vacant 
(since its vertical distance to $S$ is at least one cell's worth, guaranteed by the presence of $L_i$),
so property $\P3(i)$ holds for $y$. All other net properties remain unaffected. 
Consequently, the Invariant $\I(i+1)$ continues to hold. 

Assume now that $Y$ is a beam in the partition of $\O_{i-1}$ (i.e., $y$ is parallel to / coplanar with $R_{i-1}$). 
Since $Y$ is not part of the $(i-1)$-bridge (by case assumption), it sits on top of a visited $(i-1)$-band cell in the unfolding net. 
If $Y$ is anchored on the $i$-band, then $Y$ appears above $y$ in the unfolding net. Otherwise, 
both anchors $y$ and $y'$ of $Y$ are on the $(i-1)$-band. By~\cref{lem:front-allback}, $y'$ is visited. 
Furthermore, the $(i-1)$-band cell opposite $R_{i-1}$ must be visited, implying that $y' \neq L_{i-1}$. 
In this situation, since $y$ is visited after $y'$, 
the unfolding algorithm places $Y$ above $y$ in the unfolding net, as outlined in~\cref{sec:top}. By the net property $\P2c(i-1)$ guaranteed by the Invariant $\I2(i)$, $Y$ is up-vacant. 

If $Y$ is not adjacent to $S$, then $B$ is non-empty and at least as long as $Y$, as depicted in~\cref{fig:unfops-common-2}(b). 
Since $Y$ is up-vacant and no longer than the vertical space between $y$ and $S$, it does not overlap with $S$. Therefore, the result is a single, non-overlapping, connected net. 
At this point, invariant $\I2(i)$ no longer holds for $y$ --- however, the subsequent unfolding process for $\O_i$ does not interact with this particular net piece. 

If $Y$ is adjacent to $S$, then attaching $S$ to both $Y$ and $\ell$ reconnects the two components obtained by cutting the $(i-1)$-band along the side shared by 
$y$ and $R_{i-1}$. This results is a single  connected component. This operation merely shifts the right net component one unit to the right, therefore no overlap occurs.
Regarding relocation, if $Y$ is anchored on $R_i$, then property $\P6a(i)$ claimed by the Invariant $\I6(i)$ is applicable to $Y$. 
However, the $(i-1)$-face cell $\gamma$ located at the top of $Y$ in the unfolding net may no longer be relocatable 
if it is the only part of $Y$ adjacent to the newly attached strip $S$. However, this is not an issue, as  
the subsequent unfolding process for $\O_i$ does not interact with this particular net piece. 
All other net properties remain unaffected. 

Overall, this unfolding process preserves the local invariant $\I(i+1)$. 

\subsubsection{Proof of Correctness --- Stage 4, Common scenario 2. }
The context for common scenario 2 is that $\ell$ sits on top of $L_i$ in the unfolding net, $B$ is empty or on top of $\O_{i-1}$, the $(i-1)$-bridge consists of multiple beams---one of which is $B$---and $T$ is the $(i-1)$-bridge beam adjacent to $B$. 

The unfolding procedure and its correctness for the case where $T$ is adjacent to $S$ is identical to that of unfolding scenario 1 for this context. 
We now examine the scenario where $T$ is not adjacent to $S$. Refer to~\cref{fig:unfops-common-4}. 
In this case, $S$ gets placed in the vertical space above $T$. 
%
Let $g$ and $h$ be $T$'s front and back anchors, respectively. 
Since $T$ is not adjacent to $S$, by the orthogonal convexity of $\O_i$ and the $2$-manifold property of $\O$, both $g$ and $h$ on the $(i-1)$-band. 

If $R_{i-1}$ is a front cell, then $(i-1)$-pointer = $ccw$ (same as the $i$-pointer). 
By the net property $\P4(i-1)$ guaranteed by the Invariant $\I4(i)$, $T$ is up-vacant. 
Consequently, $S$ fits in the space above $T$ without overlap.
At this point, the Invariant $\I4(i)$ no longer holds for $T$ --- however, the subsequent 
unfolding process for $\O_i$ does not interact with this particular net piece. 

If $R_{i-1}$ is a back cell, then $(i-1)$-pointer = $cw$ (opposite to the $i$-pointer). 
Assume first that $h \neq R_{i-1}$, as depicted in~\cref{fig:unfops-common-4}(a). 
In this case, we show that the net property $\P5(i-1)$ from the Invariant $\I5(i)$ applies to $T$, 
indicating that $T$ is up-vacant. 
By~\cref{lem:bridge-visited}, all front and back $(i-1)$-band cells are visited. 
This implies that $R_{i-1}$ and $L_{i-1}$ are connected by a straight 
$(i-1)$-band segment $U$, which may be empty if $R_{i-1}$ and $L_{i-1}$ are adjacent. 
Our context is that $T$ is not adjacent to $S$, so the $(i-1)$-band segment extending $cw$ from $h$ is a left segment. 
This, combined with the orthogonally convex property of $\O_{i-1}$
and the fact that $R_{i-1}$ is a back cell, indicates that $U$ is a left segment, if non-empty. 
Consequently, if $R_{i-1}$ and $L_{i-1}$ are not adjacent, then $(i-1)$-\beam($R_{i-1}$) cannot be incident to $L_{i-1}$.
Therefore, the net property $\P5(i-1)$ applies to $T$, telling us that $T$ is up-vacant. 
Since $T$ is not adjacent to $S$, it does not overlap $S$ in the net.  
At this point, the Invariant $\I5(i)$ no longer holds for $T$ --- however, the subsequent unfolding process for $\O_i$ does not interact with this particular net piece.

Next we discuss the case where $h = R_{i-1}$, depicted in~\cref{fig:unfops-common-4}(b).  In this case, $S$ occupies the vertical space above $R_{i-1}$ in the unfolding net, so we must ensure that $R_{i-1}$ is up-vacant. 
%
We first show that the $(i-1)$-band is nearly-visited. Since $T$ is not adjacent to $S$, the straight $(i-1)$-band segment $U$ 
extending $cw$ from $R_{i-1}$ is a left segment. By the orthogonally convex property of $\O_i$ and the 2-manifold property of $\O$'s surface, $U$ must be adjacent to $S$ along a shared grid edge. Thus, there is $q \in U$ adjacent to $S$, with $(i-1)$-\beam($q$) empty. By~\cref{def:ri}, $q$ is a candidate for $R_{i-1}$ and is thus visited. This implies that $L_{i-1} \in U$, therefore the $(i-1)$-band is nearly-visited.. 

In this case, the unfolding algorithm described in~\cref{sec:bottom} partitions the bottom of $\O_{i-1}$ into beams parallel to $(i-2)$-\beam($L_{i-1}$), which are orthogonal to $(i-2)$-\beam($R_{i-1}$). This implies that $(i-2)$-\beam($R_{i-1}$) is not a beam in the partition of $\O_{i-1}$. By property $\P1a(i-1)$ guaranteed by Invariant $\I1(i)$, $R_{i-1}$ is up-vacant.  Consequently, $S$ fits above it without overlap. At this point, Invariant $\I1(i)$ no longer holds for $R_{i-1}$ --- however, the subsequent unfolding process for $\O_i$ does not interact with this particular net piece. Also note that, since $T$ is part of the $(i-1)$-bridge, property $\P6(i)$ is not applicable to $T$. All other net properties remain unaffected. 
In all cases, Invariant $\I(i+1)$ continues to hold. 
%


\subsubsection{Proof of Correctness --- Stage 4, Common scenario 3. }  
%
%
The context for common scenario 3 is that $\ell$ sits on top of $L_i$ in the unfolding net, $B$ is non-empty and on the bottom of $\O_{i}$, 
the $(i-1)$-bridge is $B$, and $f = R_{i-1}$. 
Note that, since $y$ is visited just before $R_{i-1}$, the Invariant $\I2(i)$ applies to $y$. 
We now show that $y$ is up-vacant, so $S$ fits above it without overlap.

Assume first that $Y$ is non-empty and on top of $\O_{i-1}$. See~\cref{fig:unfops-common-42}(b). 
If $y$ were a left cell, then $Y$ would be empty, contradicting our assumption. 
If $y$ were a right cell, the neighborhood around a top vertex of $y$ would not be homeomorphic to a disc, contradicting the 
$2$-manifold property of $\O$. It follows that $y$ is a front cell. This further implies that $Y$ is a beam in the partition of $\O_{i-1}$ adjacent to $S$. 
Arguments similar to those used in common scenario 1 show that 
the unfolding procedure produces a connected, non-overlapping net satisfying Invariant $\I(i+1)$. 

Assume now that $Y$ is empty or on the bottom of $\O_i$. 
Recall that in our context, the $i$-band is not nearly-visited.
Therefore, if $Y$ is empty, then $y$ meets the conditions of the net property $\P2(i-1)$ guaranteed Invariant $\I2(i)$, 
which tells us that $y$ is up-vacant. 
Consider now the case where $Y$ is non-empty. If $y$ is a left cell (as depicted in~\cref{fig:unfops-common-1}(a)) or a front cell (as depicted in~\cref{fig:unfops-common-1}(b)), then $Y$ cannot cross the $(i-1)$-bridge $B$. 
If $y$ is a right cell (as depicted in~~\cref{fig:unfops-common-1}(c)), then the portion of $Y$ not overlapping $B$ is either empty or consists of a single component anchored on the $i$-band. In either of these case, $Y$ is stable (cf.~\cref{def:stable}). This, along with the fact that the $i$-band is not nearly-visited, indicates that the net property $\P2b(i-1)$ guaranteed by Invariant $\I2(i)$ applies to $y$. Therefore, $y$ is up-vacant. 

At this point, Invariant $\I2(i)$ no longer holds for $y$ --- however, the subsequent unfolding process for $\O_i$ does not interact with this particular net piece. 

Regarding relocation, the beam $B$ used to attach $S$ is the $(i-1)$-bridge, therefore the net property $\P6(i)$ is not applicable to $B$. 
All other net properties remain unaffected. Consequently, the Invariant $\I(i+1)$ continues to hold.

\subsubsection{Proof of Correctness --- Stage 4, Common scenario 4. }  
The context for common scenario 4 is that $\ell$ does not sit above $L_i$ in the unfolding net, $i$-\beam($L_i$) is part of the $i$-bridge, $L_i$ and $R_i$ are connected by $S$, and $z \in S$ is adjacent to $R_i$. See~\cref{fig:unfops-common-5}. Let $Z = (i-1)$-\beam($z$). 

\paragraph{Case $Z$ empty} Refer to~\cref{fig:unfops-common-5}(a). 
By the net property $\P3(i-1)$ guaranteed by Invariant $\I3(i)$, $z$ attaches above $u$ without overlap. If $S = z$, the unfolding is complete. Assume now that $S \setminus z$ is non-empty. 

By the orthogonally convex property of $\O_i$, $R = (i-1)$-\beam($R_i$) is a bottom beam anchored on two visited $i$-band cells, 
namely $R_i$ and $h$. 
In such a case, our the unfolding algorithm attaches $R$ to the second visited $i$-band cell, which is $R_i$. 
Since $R$ is parallel to the $(i-1)$-bridge, it cannot cross it. 
Since $L_i$ and $R_i$ are parallel, the $i$-band is not nearly-visited. 
By the net property $\P1b(i)$ guaranteed by Invariant $\I1(i)$, $h$ is down-vacant. 
By the net property $\P6b(i)$ guaranteed by Invariant $\I6(i)$, $R$ is relocatable. 
These together enable us to disconnect $R$ from $R_i$ and reattach it below $h$ in the unfolding net 
without breaking connectivity or causing overlap, thereby rendering $R_i$ down-vacant. As a result, 
$S \setminus z$ can attach to the $i$-bridge to land below $R_i$ without overlap. 

At this point, $R$ is no longer relocatable and Invariant $\I6(i)$ no longer holds for $R$.
Also, Invariant $\I1(i)$ no longer holds for $h$. 
The unfolding of $\O_i$, however, is complete.
Furthermore, since the $i$-band is not nearly-visited, invariant $\I1(i+1)$ is not applicable to $h$. 
All other net properties remain unaffected. 
Consequently, the invariant $\I(i+1)$ continues to hold. 
 
\paragraph{Case $Z$ on the bottom of $\O_i$} Refer to~\cref{fig:unfops-common-5}(b). 
Since $Z$ is on the bottom of $\O_i$, $G = (i-1)$-\beam($g$) is also on the bottom of $\O_i$ and anchored on the $(i-1)$-band. Our unfolding algorithm attaches $G$ below $g$ (see~\cref{sec:bottom}), and $S$ to $G$. 
Once attached to $G$, the segment $S$ occupies the space below $h$, which---
as shown in the previous case when $Z$ is empty---is down-vacant. Thus, no overlap occurs.

Regarding relocation, the beam $G$ used to attach $S$ is not solely anchored on the $i$-band, so the net property 
$\P6b(i)$ is not applicable to $G$. This, together with arguments from the previous case, shows that the invariant 
$\I(i+1)$ continues to hold. 

\paragraph{Case $Z$ on the top of $\O_{i-1}$} Refer to~\cref{fig:unfops-common-5}(c). 
Since $R_i$ is not adjacent to $L_i$, the net property $\P1a(i)$ guaranteed by Invariant $\I1(i)$ applies to $R_i$, indicating that 
$R_i$ is down-vacant. Consequently, $x$ can be attached below $R_i$ without overlap.  We now show that $x$ is relocatable. 

Let $T = (i-1)$-beam($R_i$) be the top beam containing $x$, and let $u$ be $T$'s front anchor. By~\cref{lem:beam-ends}, $u$ is on the $(i-1)$-band. 
Note that $T$ may not be part of the $(i-1)$-bridge; otherwise, $R_i$ would be adjacent to the $(i-1)$-bridge and therefore a candidate for $L_i$, contradicting~\cref{def:ri}. Also, note that $T$ is in the partition of $\O_{i-1}$, since $R_i$ is parallel to $L_i$.
Our unfolding algorithm positions $T$ on top of $u$ (see~\cref{sec:top}).

If $T \setminus x$ is non-empty, or if $T$ is not anchored on $L_{i-1}$, then by 
the net property $\P6a(i)$ guaranteed by Invariant $\I6(i)$, $x$ is relocatable. 
This allows us to relocate $x$ below $R_i$, without disconnecting the net and without overlap. We now show that this is 
the only possible scenario. 

\vspace{1em}
\begin{figure}[htp]
    \centering
    \includegraphics[page=6,width=0.6\textwidth]{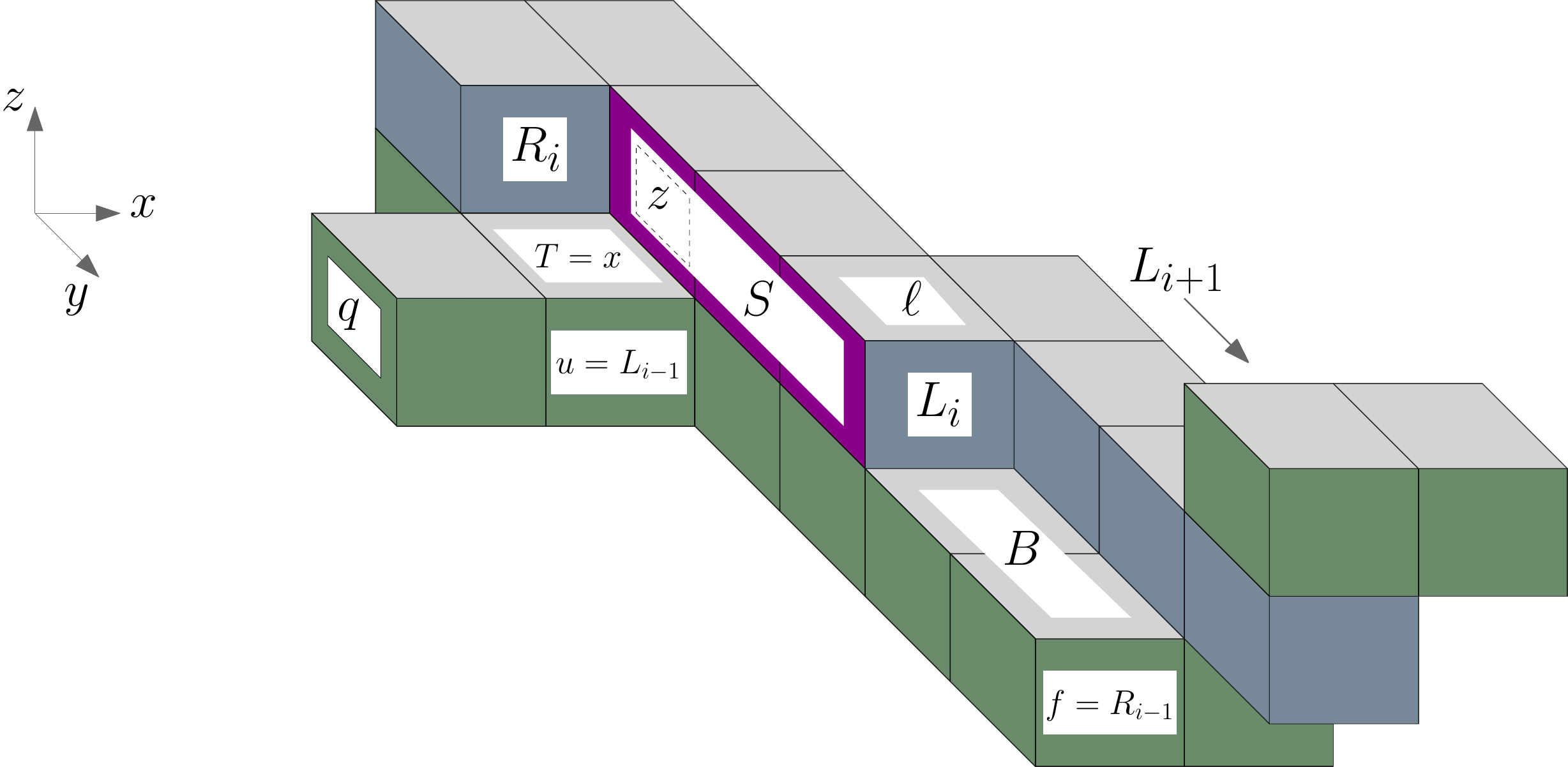}
    \caption{Common scenario 4.3 correctness: $x = (i-1)$-\beam($R_i$) is anchored on $L_{i-1}$.}
    \label{fig:unfops-common-6}
\end{figure}

\vspace{1em}
Assume, for the sake of contradiction, that $T \setminus x$ is empty and $u = L_{i-1}$ (see~\cref{fig:unfops-common-6}).  
%
Note that the  $(i-1)$-face $F$ containing $T$ also includes  $(i-1)$-\beam($z$), which, by~\cref{lem:beam-ends},
is anchored on an $(i-1)$-band cell $q$. By~\cref{lem:liri-out}, $q$ is visited. 
Since $T$ consists of a single cell (namely $x$) and $S$ is a non-empty left segment, 
$F$ does not include $B$. 
Moreover, since $B$ is part of the $(i-1)$-bridge (cf.~\cref{def:ri}), it follows that 
$F$ does not include the $(i-1)$-bridge. Consequently, $R_{i-1}$ is not adjacent to $F$. 
Given that $q$ is visited and adjacent to $F$, we conclude that the $(i-1)$-pointer is $cw$.  
This implies that $f$, the last visited $(i-1)$-band cell parallel to $L_i$, must be $R_{i-1}$ (cf.~\cref{def:ri}). 
However, in this case, both $R_{i-1}$ and $L_i$ are front faces, implying that the $i$-pointer is also $cw$ (same as the $(i-1)$-pointer), 
a contradiction. 

We have shown that the net property $\P6a(i)$ applies to $T$ and thus $x$ is relocatable. 
The subsequent unfolding procedure that cuts the $i$-bridge and inserts $S$ between the two resulting components 
merely shifts the right net component one unit to the  right and does not affect the local invariant. 

After relocating $x$, Invariant $\I6(i)$ no longer holds for $T$ and $\I1(i)$ no longer holds for $R_i$ --- however, 
the unfolding of $\O_i$ is now complete. 
Furthermore, since the $i$-band is not nearly-visited, invariant $\I1(i+1)$ does not apply to $R_i$. 
All other properties remain unaffected, so invariant $\I(i+1)$ continues to hold.  

\subsubsection{Proof of Correctness --- Stage 4, Case 3. }  
This section references the notation used in~\cref{sec:case3}. 
The context for this case is that both $a$ and $b$ are visited. 
Since $i$-beam($S$) is on top of $\O_i$, the $i$-bridge must be on top of $\O_i$. 
%
In particular, $i$-\beam($R_i$) sits on top of $R_i$ in the unfolding net and is adjacent to $S$ along its right side. 
Refer to~\cref{fig:unfops-4}. 

\vspace{1em}
\begin{figure}[htp]
    \centering
    \includegraphics[page=4,width=0.99\textwidth]{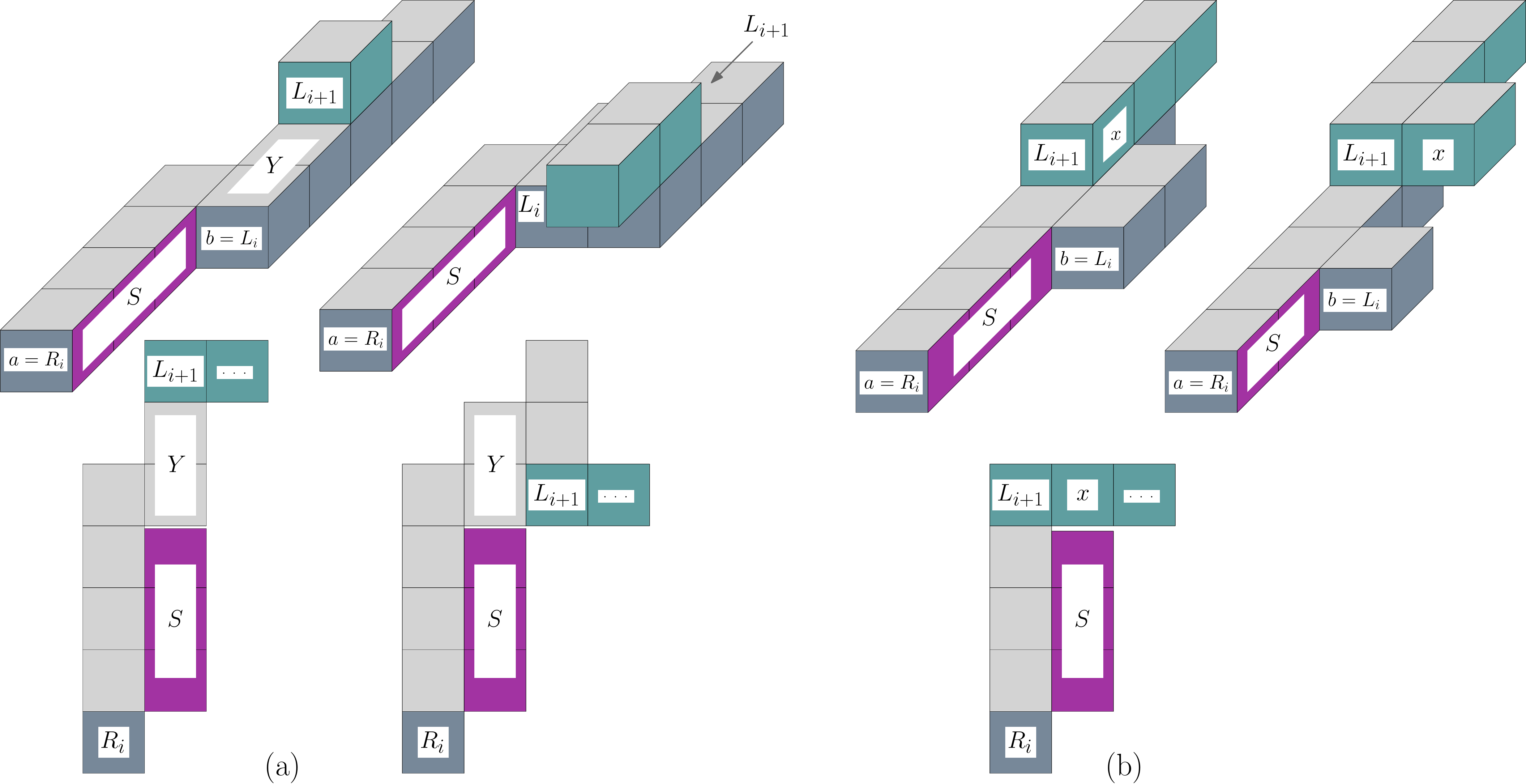}
    \caption{Case 3: $S$ is a \emph{right} segment and the $i$-bridge consists of (a) multiple beams (b) a single beam.}
    \label{fig:unfops-4}
\end{figure}

\vspace{1em}
Assume first that the $i$-bridge consists of two or more beams, and let $Y$ be the $i$-bridge beam adjacent to $i$-\beam($R_i$). Since $Y$ is not anchored on $R_{i}$, the net property $\P5(i)$ guaranteed by invariant $\I5(i)$ tells us that $Y$ is down-vacant. 
Our unfolding procedure attaches $S$ to the right side of $i$-\beam($R_i$), positioning it just below $Y$ in the unfolding net, therefore no overlap occurs. 

At this point, the Invariant $\I5(i)$ no longer holds for $Y$ --- however, the subsequent unfolding process for $\O_i$ does not interact with this net piece. Furthermore, since $L_i$ and $R_i$ are not adjacent (because $S$ is non-empty) and $i$-\beam($R_i$) is incident to $L_i$, Invariant $\I5(i+1)$ is not applicable to $Y$. 
Regarding relocation, the beam $i$-\beam($R_i$) used to attach $S$ is part of the $i$-bridge, so net property $\P6(i)$ does not apply to 
it. All other net properties remain unaffected. 
Consequently, the invariant $\I(i+1)$ continues to hold. 

Assume now that the $i$-bridge consists of a single beam, namely $i$-\beam($R_i$). In this case the $(i+1)$-pointer is $ccw$ (same as $i$-pointer). Refer to~\cref{fig:unfops-4}(b). This case mirrors Case 2 (\cref{sec:case2},~\cref{fig:unfops-2}). It is important to note that the assumption of $b$ being unvisited in Case 2 is solely used to establish that the $i$-bridge consists of a single beam, which aligns with the conditions of this case. Arguments identical to the ones used in Case 2 show that no overlap occurs and Invariant $\I(i+1)$ continues to hold.

\subsubsection{Proof of Correctness --- Stage 4, Case 4. }  
This section references the notation used in~\cref{sec:case4}. The context for this case is that $a$ is unvisited and $b$ is visited. 
Recall that $a'$ be the $i$-band cell opposite $a$. 
By~\cref{lem:front-allback}, $a'$ is visited.   

\vspace{1em}
\begin{figure}[htp]
    \centering
    \includegraphics[page=7,width=0.88\textwidth]{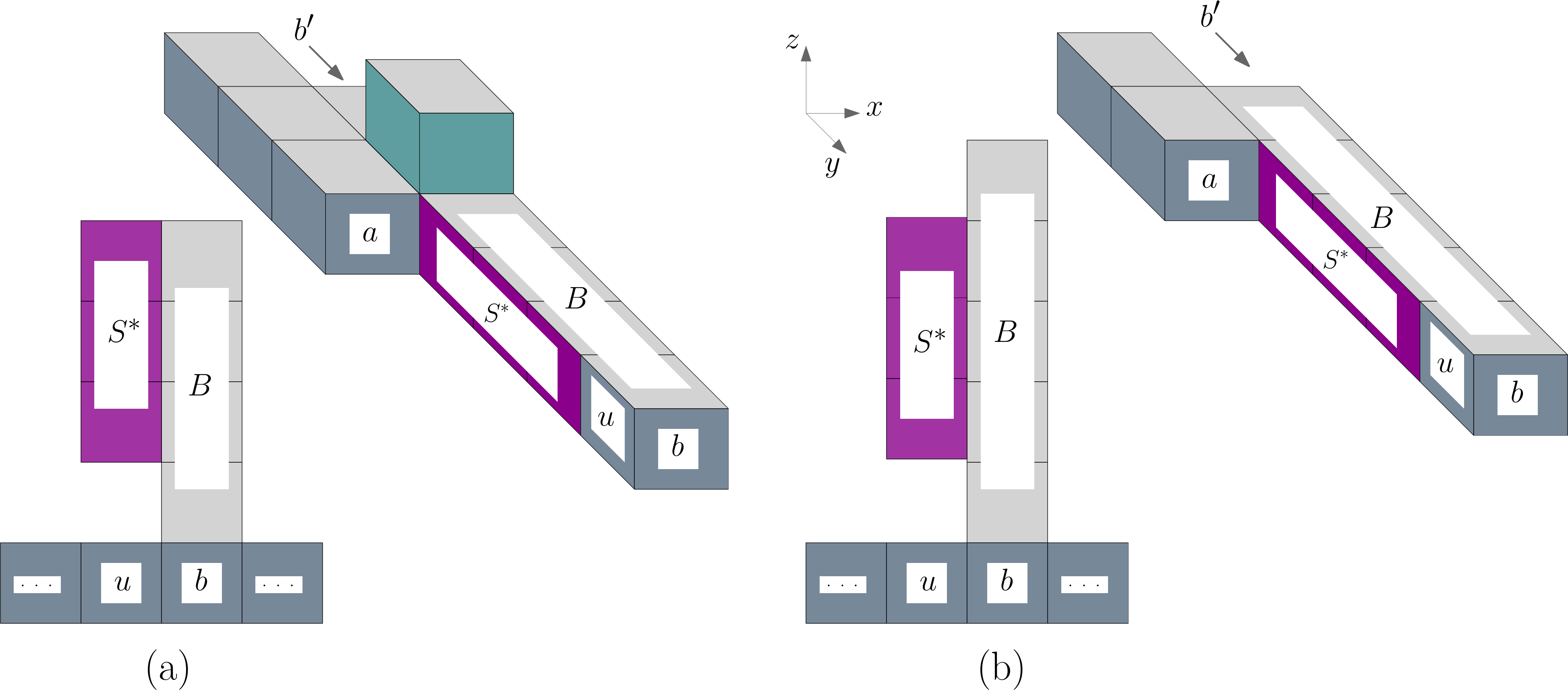}
    \caption{Case 4: $S$ is a \emph{left} segment and $B$, $B'$ are (a) distinct (b) identical.}
    \label{fig:unfops-7}
\end{figure}

\vspace{1em}
Consider first the case where $S$ is a left segment. Refer to~\cref{fig:unfops-7}. 
The case where $b = L_i$ matches the setting for the common scenario, whose correctness has already been established. 
%
Suppose now that $S$ includes $L_i$. 
By~\cref{lem:liri-out}, $i$-beam($S^*$) is non-empty, lies on top of $\O_i$, and is anchored on the $i$-band only. 
This implies that $S^*$ is adjacent along its top side to either $B = i$-\beam($b$) or $B' = i$-\beam($b'$). 

We  show that $S^*$ is adjacent to $B$. Assuming the opposite, $B'$ must be adjacent to $S^*$, anchored on an $(i+1)$-band cell. This implies that $i$-\clip($a$) includes $B'$ and is adjacent to an $(i+1)$-band cell parallel to $a$, contradicting~\cref{lem:liri-out}.  
Therefore, $S^*$ must be adjacent to $B$. 

If $B \neq B'$ (as shown in~\cref{fig:unfops-7}(a)), then $B$ must be anchored on an $(i+1)$-band cell. In this case, the unfolding algorithm places $B$ above $b$ in the unfolding net (see~\cref{sec:top}). 

Assume now that $B = B'$, as shown in~\cref{fig:unfops-7}(b). 
Note that $B$ cannot be part of the $i$-bridge -- otherwise, the $i$-bridge would consist of multiple beams, including $i$-\beam($a$),  
and $a$ would be visited (cf.~\cref{lem:bridge-visited}), contradicting this case assumption. 
This, along with the fact that $b'$ is visited after $b$ (since $L_i$ is part of $S$), 
implies that $B$ appears above $b'$ in the unfolding net (see~\cref{sec:top}). 
By the net property $\P6c(i)$ guaranteed by Invariant $\I6(i)$, the beam $B$ is relocatable, 
so relocating it to sit on top of $b$ does not disconnect the net. 


At this point, it is guaranteed that $B$ appears above $b$ in the unfolding net. 
%
Let $u \in S$ to be the $i$-band cell adjacent to $b$ in the $cw$ direction. Since $S$ includes $L_i$, $u$ is visited and appears
left of $b$ in the unfolding net. Since $u$ is orthogonal to $R_i$ (which is parallel to $a$ by assumption), 
$i$-\beam($u$) is not in the partition of $\O_i$ (also note that it is non-empty, since it is part of $B$). By the net property $\P2a(i)$ guaranteed by Invariant $(\I2)$, $u$ is up-vacant. Consequently, $S^*$ fits in this vacant space without overlap. 

At this point, the Invariant $\I2(i)$ no longer holds for $u$ and $\I6(i)$ no longer holds for $B$ --- however, the subsequent unfolding 
process for $\O_i$ does not interact with this net piece. Furthermore, since $b \neq R_i$, the Invariant $\I2(i+1)$ is not applicable to $u$. 
Thus, Invariant $\I(i+1)$ continues to hold. 

\vspace{1em}
\begin{figure}[htp]
    \centering
    \includegraphics[page=6,width=0.95\textwidth]{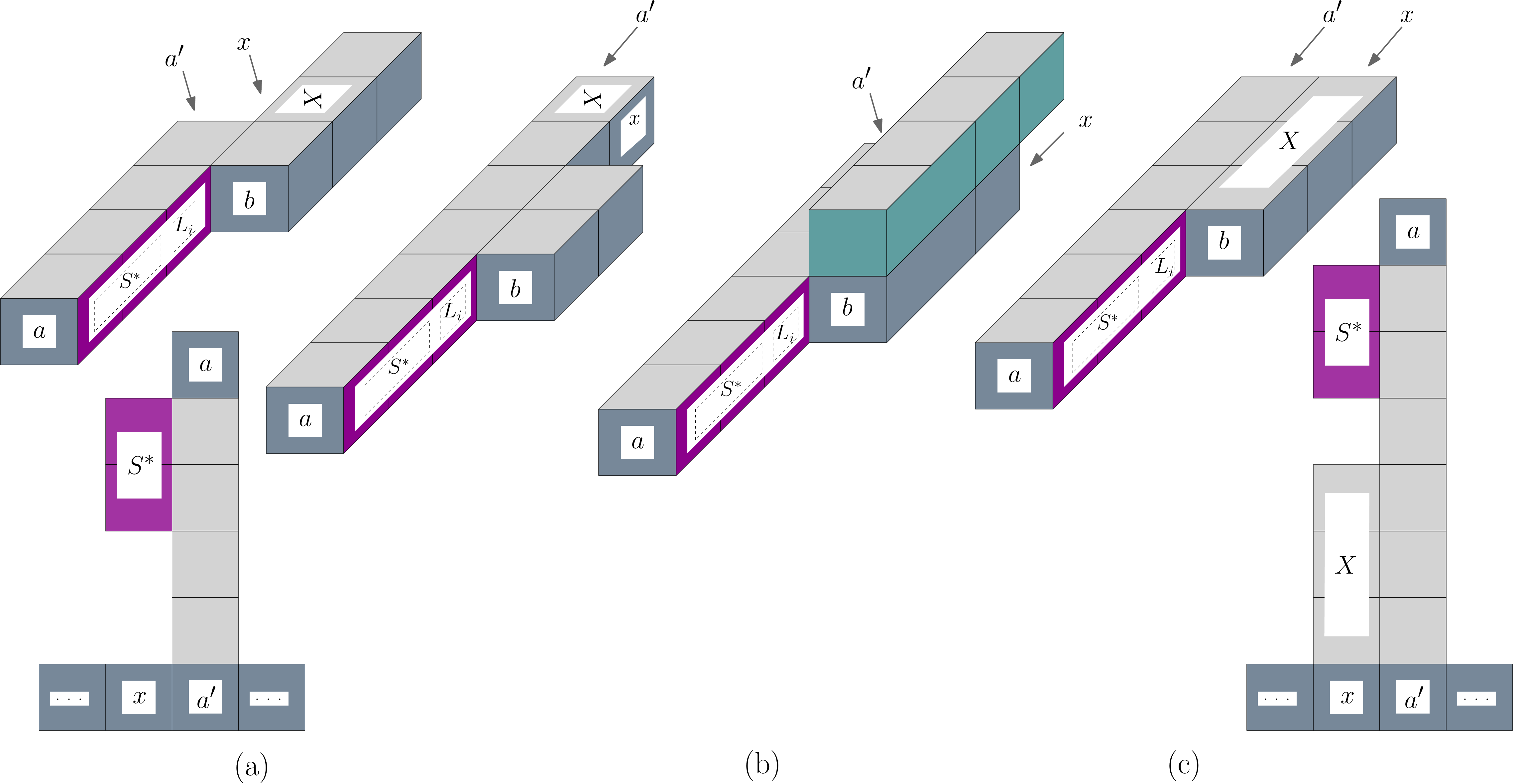}
    \caption{Case 4: $S$ is a \emph{right} segment and $b$ is a front cell.}
    \label{fig:unfops-6}
\end{figure}

\vspace{1em}
Consider now the case where $S$ is a right segment. In this case, $b$ may be a front of a back cell. 
Assume first that $b$ is a front cell. Refer to~\cref{fig:unfops-6}. 
This scenario is a horizontal reflection of the setup in Case 2 (see~\cref{fig:unfops-3}). The only difference is that $i$-\beam($b$), which corresponds to the horizontal reflection of  $i$-\beam($R_i$) in~\cref{fig:unfops-3}, is not the $i$-bridge in this case. 
Using the same notation, let $x$ is the $i$-band cell $cw$ from $a'$, and let $X = i$-\beam($x$). 
The proof of correctness for the cases where $X$ is empty or $x$ is orthogonal to $R_i$ (and thus to $a$ and $b$) is independent 
of the presence of the $i$-bridge and applies here as well. For the case where $X$ is non-empty and $x$ is parallel to $R_i$, 
the following two observations allow the proof of correctness to remain the same: 
\begin{itemize}
\squeezelist
\item If $X$ is on the bottom of $\O_{i+1}$ (as depicted in~\cref{fig:unfops-6}(b)), we must show that, if the $(i+1)$-band is nearly-visited, then $X$'s anchor on the $(i+1)$-band is visited. 
Since $a$ is unvisited, $i$-\clip($a$) is not adjacent to an $(i+1)$-band rectangle parallel to $a$ (cf.~\cref{lem:liri-out}). This, along with~\cref{lem:opposite-edges}, implies that 
$i$-\beam($b$) is empty and $\O_{i+1}$ is adjacent to $i$-\beam($a$) along a left $(i+1)$-band segment $U$. Since $a'$ is visited, $R_i$ lies $ccw$ from $a'$. 
These together  imply that $U$ includes $L_{i+1}$. Consequently, if the $(i+1)$-band is nearly-visited, $X$'s anchor on the $(i+1)$-band is precisely $R_{i+1}$ and is, therefore, visited. 
\item If $X$ is on top of $\O_i$ (as depicted in~\cref{fig:unfops-6}(c)), then $X$ cannot be anchored on $L_{i+1}$. Otherwise, 
$i$-\clip($a$), which includes $X$, would be adjacent to $L_{i+1}$, rendering $a$ visited (cf.~\cref{lem:liri-out}), a contradiction.
\end{itemize}
These observations allow the correctness proof from Case 2 to carry over to this scenario.

\setlength{\intextsep}{-0.2em}%
\setlength{\columnsep}{1em}%
\begin{wrapfigure}{r}{0.43\textwidth}
\raisebox{-1.5em}{\includegraphics[page=8,width=0.43\textwidth]{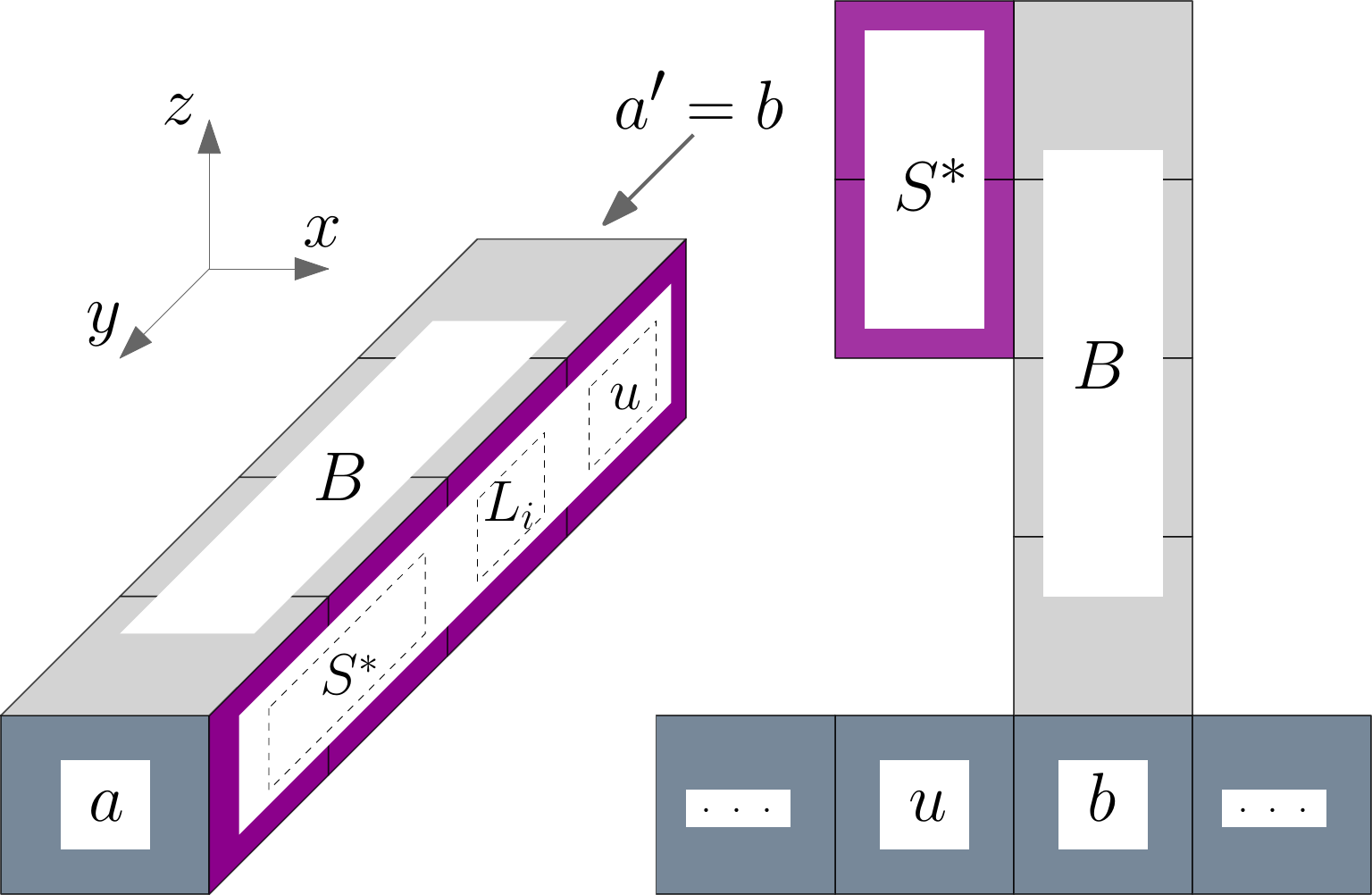}}
 \caption{Case 4: $S$ is a right segment and $b$ is a back cell.}
\label{fig:unfops-8}
\end{wrapfigure}

Assume now that $b$ is a back cell, as depicted in~\cref{fig:unfops-8}. Since $b$ is the first $i$-band cell parallel to $a$ encountered in a 
$ccw$ walk from $a$, we have $a' = b$. 
If $L_i \in S$, then the $i$-band cell $cw$ from $b$  is visited and the unfolding process for $S^*$ matches the one depicted in~\cref{fig:unfops-6}: 
$S^*$ gets attached to the left side of $B = i$-\beam($b$). Arguments similar to those outlined above show no overlap occurs and the Invariant $\I(i+1)$ continues to hold. 

\medskip
If $L_i \notin S$, $b = L_i$ (since $b$ is visited) and we are in the common scenario, whose correctness has already been established. 
This concludes the proof of correctness for the unfolding algorithm.

%% file: unf-ortho.bbl
\begin{thebibliography}{10}
\providecommand{\url}[1]{\texttt{#1}}
\providecommand{\urlprefix}{URL }
\providecommand{\doi}[1]{https://doi.org/#1}

\bibitem{Bern-Demaine-Eppstein-Kuo-Mantler-Snoeyink-2003}
Bern, M., Demaine, E., Eppstein, D., Kuo, E., Mantler, A., Snoeyink, J.:
  Ununfoldable polyhedra with convex faces. Computational Geometry: Theory and
  Applications  \textbf{24}(2),  51--62 (February 2003)

\bibitem{BDDLOORW1998}
Biedl, T., Demaine, E., Demaine, M., Lubiw, A., Overmars, M., O'Rourke, J.,
  Robbins, S., Whitesides, S.: Unfolding some classes of orthogonal polyhedra.
  In: Proc. of the 10th Canadian Conference on Computational Geometry. pp.
  70--71. Montr\'eal, Canada (August 1998)

\bibitem{CST-2012}
Chambers, E.W., Sykes, K., Traub, C.M.: Unfolding rectangle-faced orthostacks.
  In: Proc. of the 24th Canadian Conference on Computational Geometry. pp.
  23--28 (2012)

\bibitem{Chang2015}
Chang, Y.J., Yen, H.C.: Unfolding orthogonal polyhedra with linear refinement.
  In: Proc. of the 26th International Symposium on Algorithms and Computation,
  ISAAC 2015, Nagoya, Japan. pp. 415--425. Springer Berlin Heidelberg (2015)

\bibitem{Damian-Demaine-Flatland-2014-delta}
Damian, M., Demaine, E., Flatland, R.: Unfolding orthogonal polyhedra with
  quadratic refinement: the {D}elta-unfolding algorithm. Graphs and
  Combinatorics  \textbf{30}(1),  125--140 (2014)

\bibitem{Damian-Demaine-Flatland-2017-O'Rourke-genus2}
Damian, M., Demaine, E., Flatland, R., O'Rourke, J.: Unfolding genus-2
  orthogonal polyhedra with linear refinement. Graphs and Combinatorics
  \textbf{33}(5),  1357--1379 (Sep 2017)

\bibitem{Damian-Flatland-21}
Damian, M., Flatland, R.: Unfolding polycube trees with constant refinement.
  Computational Geometry  \textbf{98},  101793 (2021)

\bibitem{Damian-Flatland-Meijer-O'Rourke-2005-orthotrees}
Damian, M., Flatland, R., Meijer, H., O'Rourke, J.: Unfolding well-separated
  orthotrees. In: Abstracts from the 15th Annual Fall Workshop on Computational
  Geometry. pp. 23--25 (November 2005)

\bibitem{Damian-Flatland-O'Rourke-2005-manhattan}
Damian, M., Flatland, R., O'Rourke, J.: Unfolding {M}anhattan towers. In: Proc.
  of the 17th Canadian Conference on Computational Geometry. pp. 211--214.
  Windsor, Canada (August 2005)

\bibitem{Damian-Flatland-O'Rourke-2007-epsilon}
Damian, M., Flatland, R., O'Rourke, J.: Epsilon-unfolding orthogonal polyhedra.
  Graphs and Combinatorics  \textbf{23}(1),  179--194 (2007)

\bibitem{Damian-Meijer-2004-orthostacks}
Damian, M., Meijer, H.: Edge-unfolding orthostacks with orthogonally convex
  slabs. In: Abstracts from the 14th Annual Fall Workshop on Computational
  Geometry. pp. 20--21. Cambridge, MA (November 2004)

\bibitem{DM24}
Damian, M., Meijer, H.: Edge-unfolding polycubes with orthogonally convex
  layers (2024), \url{https://arxiv.org/abs/2407.01326}

\bibitem{Demaine-O'Rourke-2007}
Demaine, E., O'Rourke, J.: Geometric Folding Algorithms: Linkages, Origami,
  Polyhedra. Cambridge University Press (July 2007)

\bibitem{DK-2023}
Demaine, E.D., Karntikoon, K.: Unfolding orthotubes with a dual hamiltonian
  path. Thai Journal of Mathematics  \textbf{21}(4),  1011--1023 (2023)

\bibitem{open-problems}
Demaine, E.D., Mitchell, J.S.B., O’Rourke, J.: The Open Problems Project
  \url{https://topp.openproblem.net}

\bibitem{Liou-Poon-Wei-2014-onelayer}
Liou, M.H., Poon, S.H., Wei, Y.J.: On edge-unfolding one-layer lattice
  polyhedra with cubic holes. In: The 20th International Computing and
  Combinatorics Conference (COCOON) 2014. pp. 251--262 (2014)

\bibitem{Josef-22}
Minar\'ic, J.: Unfolding some classes of one-layer polycubes. In: Proc. of the
  34th Canadian Conference on Computational Geometry. pp. 243--248. Toronto,
  Canada (August 2022)

\bibitem{Klara-2024}
Pernicov\'a, K.: Grid-edge unfolding orthostacks with rectangular slabs. In:
  Proc. of the 36th Canadian Conference on Computational Geometry. pp. 65--70
  (2024)

\end{thebibliography}
